\documentclass[onecolumn]{IEEEtran}
\usepackage{algorithm}
\usepackage{algpseudocode}
\usepackage{times,amssymb,amsmath,amsfonts,nicefrac,float,accents,color,
graphicx,bbm,caption,subcaption,mathrsfs,amsthm,stmaryrd,soul}
\usepackage[noadjust]{cite}
\usepackage{tikz}
\usetikzlibrary{calc}
\usepackage[mathscr]{eucal}
\usepackage{arydshln}

\newtheorem{theorem}{Theorem}[section]
\newtheorem{lemma}[theorem]{Lemma}
\newtheorem{proposition}[theorem]{Proposition}

\newtheorem{definition}[theorem]{Definition}

\newtheorem{remark}[theorem]{Remark}

\newtheorem{fact}[theorem]{Fact}
\newtheorem*{MDS}{MDS Conjecture}

\newcommand{\tb}{\textbf}

\begin{document}
\title{New Bounds on the Field Size for Maximally Recoverable Codes Instantiating Grid-like Topologies}
\author{Xiangliang Kong, Jingxue Ma and Gennian Ge
\thanks{The research of G. Ge was supported by the National Natural Science Foundation of China under Grant Nos.11431003 and 61571310, Beijing Scholars Program, Beijing Hundreds of Leading Talents Training Project of Science and Technology, and Beijing Municipal Natural Science Foundation.}
\thanks{X. Kong is with the School of Mathematical Sciences, Capital Normal University, Beijing 100048, China (e-mail: 2160501011@cnu.edu.cn).}
\thanks{J. Ma is with the School of Mathematical Sciences, Zhejiang University, Hangzhou 310027, China (email: majingxue@zju.edu.cn).}
\thanks{G. Ge is with the School of Mathematical Sciences, Capital Normal University, Beijing 100048, China (e-mail: gnge@zju.edu.cn).}
}

\date{}\maketitle

\begin{abstract}
  In recent years, the rapidly increasing amounts of data created and processed through the internet resulted in distributed storage systems employing erasure coding based schemes. Aiming to balance the tradeoff between data recovery for correlated failures and efficient encoding and decoding, distributed storage systems employing maximally recoverable codes came up. Unifying a number of topologies considered both in theory and practice, Gopalan et al. \cite{Gopalan2017} initiated the study of maximally recoverable codes for grid-like topologies.

  In this paper, we focus on the maximally recoverable codes that instantiate grid-like topologies $T_{m\times n}(1,b,0)$. To characterize the property of codes for these topologies, we introduce the notion of \emph{pseudo-parity check matrix}. Then, using the Combinatorial Nullstellensatz, we establish the first polynomial upper bound on the field size needed for achieving the maximal recoverability in topologies $T_{m\times n}(1,b,0)$. And using hypergraph independent set approach, we further improve this general upper bound for topologies $T_{4\times n}(1,2,0)$ and $T_{3\times n}(1,3,0)$. By relating the problem to generalized \emph{Sidon sets} in $\mathbb{F}_q$, we also obtain non-trivial lower bounds on the field size for maximally recoverable codes that instantiate topologies $T_{4\times n}(1,2,0)$ and $T_{3\times n}(1,3,0)$.

\smallskip

\end{abstract}


\section{Introduction}\label{secintro}

With rapidly increasing amounts of data created and processed in internet scale companies such as Google, Facebook, and Amazon, the efficient storage of such copious amounts of data has thus become a fundamental and acute problem in modern computing. This resulted in distributed storage systems relying on distinct storage nodes. Modern large scale distributed storage systems, such as data centers, used to store data in a redundant form to ensure reliability against node failures. However, this strategy entails large storage overhead and is nonadaptive for modern systems supporting the ``Big Data" environment.

To ensure the reliability with better storage efficiency, erasure coding schemes are employed, such as in Windows Azure \cite{5} and in Facebook's Hadoop cluster \cite{14}. However, in traditional erasure coding scheme, if one node fails, which is the most common failure scenario, we may recover it by accessing a large amount of the remaining nodes. This is a time consuming recovery process. To address this efficiency problem, a lot of works have emerged in two aspects: local regeneration and local reconstruction.

The concept of local regeneration was introduced by Dimakis et al. \cite{DGWR07}. They established a tradeoff between the repair bandwidth and the storage capacity of a node, and introduced a new family of codes, called regenerating codes, which attained this tradeoff. The concept of local reconstruction was introduced by Gopalan et al. \cite{Gopalan12}, and they initiated the study of Local Reconstruction Codes (LRCs). We say a certain node has \emph{locality $r$} if it can be recovered by accessing only $r$ other nodes, and LRCs are linear codes with all-symbol \emph{locality $r$}. In recent years, the theory of regenerating codes and LRCs has developed rapidly. There have been a lot of related works focusing on the bounds and the constructions of optimal codes, see \cite{WDK07,WD09,RSK11,Song2017,Tamo1408,Tamo1406,Wang1411,11,16,Prakash1406,Huang2016,Shahabinejad2016} and the reference therein.

The notion of maximally recoverable property was first introduced by Chen et al. \cite{CHL07} for multi-protection group codes, and then extended by Gopalan et al. \cite{Gopalan2014} to general settings. In \cite{Gopalan2014}, the authors introduced the \emph{topology} of the code to specify the supports for the parity check equations, and they also obtained a general upper bound on the minimal size of the field over which maximally recoverable (MR) codes exist.

Different from the parity check matrix, the topology of the code only specifies the number of redundant symbols and the data symbols on which the redundant ones depend. This makes the topology a crucial characterization of the structure of the code used under distributed storage settings. With the purpose of deploying longer codes in storage, Gopalan et al. \cite{Gopalan2017} proposed a family of topologies called \emph{grid-like topologies}, which unified a number of topologies considered both in theory and practice.

Consider an $m\times n$ matrix, each entry storing a data from a finite field $\mathbb{F}$. Every row satisfies a given set of $a$ parity constraints, and every column satisfies a given set of $b$ parity constraints. In addition, there are $h$ global parity constraints that involve all $mn$ entries from the matrix. The topology of the code under these three constraints is denoted by $T_{m\times n}(a,b,h)$. In \cite{Gopalan2017}, the authors considered the maximal recoverable codes for general grid-like topologies, and they established a super-polynomial lower bound on the field size needed for achieving maximal recoverability in any grid-like topologies $T_{m\times n}(a,b,h)$ with $a,b,h\geq 1$. They also tried to characterize correctable erasure patterns for grid-like topologies of the form $T_{m\times n}(a,b,0)$, and obtained a full combinatorial characterization for the case of $T_{m\times n}(1,b,0)$.

The general lower bound given in \cite{Gopalan2017} is obtained from the case of a basic topology $T_{n\times n}(1,1,1)$, where the lower bound requires field size $q=2^{\Omega((\log n)^2)}$. Recently, by relating the problem to the independence number of the Birkhoff polytope graph, Kane et al. \cite{KLR17} improved the lower bound to $q\geq2^{(\frac{n}{2}-2)}$ using the representation theory of the symmetric group. They also obtained an upper bound $q\leq 2^{3n}$ using recursive constructions.

As for other related works, Gandikota et al. \cite{GGTZ17} considered the maximal recoverability for erasure patterns of bounded size. Shivakrishna et al. \cite{SRLS2018} considered the recoverability of a special kind of erasure patterns called extended erasure patterns for topologies $T_{(m+m')\times n}(2,b,0)$. It is worth noting that, Gopi et al. \cite{Gopi2018} recently obtained a super-linear lower bound for maximally recoverable LRCs which can be viewed as the MR codes for topology $T_{\frac{n}{r}\times r}(a,0,h)$.

In this paper, we focus on the maximally recoverable codes that instantiate topologies of the form $T_{m\times n}(1,b,0)$, which can be regarded as tensor product codes of column codes with a single parity constraint and row codes with $b$ parity constraints. In order to describe the parity constraints globally, we introduce the notion of \emph{pseudo-parity check matrix}, which can be viewed as a generalization of the parity check matrix. Based on this, using tools from extremal graph theory and additive combinatorics, we prove the following results:
\begin{itemize}
  \item The first polynomial upper bound on the minimal size of the field required for the existence of MR codes that instantiate topologies $T_{m\times n}(1,b,0)$:
      \begin{equation*}
      q\leq C_0(m,b)\cdot n^{2b(m-1)}+n^{(b-1)},
      \end{equation*}
      where $C_0(m,b)=(m+1)!\cdot{m\cdot b(m-1)\choose \leq2b(m-1)}$;
  \item Further improved upper bounds on the field size required for MR codes instantiating topologies $T_{4\times n}(1,2,0)$ and $T_{3\times n}(1,3,0)$:
      \begin{equation*}
      q\leq\mathcal{O}(\frac{n^{5}}{\log (n)});
      \end{equation*}
  \item A polynomial lower bound on the minimal size of the field required for MR codes instantiating topologies $T_{4\times n}(1,2,0)$:
      \begin{equation*}
      q\geq\frac{(n-3)^{2}}{4}+2,
      \end{equation*}
       and a linear lower bound on the minimal size of the field required for MR codes instantiating topologies $T_{3\times n}(1,3,0)$:
      \begin{equation*}
      q\geq\frac{\sqrt{n^2-11n+34}}{2}.
      \end{equation*}
\end{itemize}

The paper is organised as follows: In Section II, we give the formal definitions for general topologies, grid-like topologies and maximal recoverability, we also include some known results for topologies $T_{m\times n}(a,b,0)$ and the tools from hypergraph independent set. In Section III, we introduce the notion of \emph{pseudo-parity check matrix} and regular irreducible erasure patterns. In Section IV, we present our proof for the general polynomial upper bound on the minimal size of the field required for the existence of MR codes that instantiate topologies $T_{m\times n}(1,b,0)$. In Section V, we improve the general upper bound for MR codes that instantiate topologies $T_{4\times n}(1,2,0)$ and $T_{3\times n}(1,3,0)$, and we also establish non-trivial lower bounds for both cases. In Section VI, we conclude our work and list some open problems.

\section{Preliminaries}\label{secpre}

\subsection{Notation}

We use the following standard mathematical notations throughout this paper.

\begin{itemize}
  \item Let $q$ be the power of a prime $p$, $\mathbb{F}_q$ be the finite field with $q$ elements, $\mathbb{F}_{q}^{n}$ be the vector space of dimension $n$ over $\mathbb{F}_q$ and $\mathbb{F}_{q}^{m\times n}$ be the collection of all $m\times n$ matrices with elements in $\mathbb{F}_q$.
  \item For any vector $\tb{v}=(v_1,\cdots,v_n)\in \mathbb{F}_{q}^{n}$, let ${\rm supp}(\tb{v})=\{i\in [n] : v_i\neq 0\}$ and ${\omega}(\tb{v})=|{\rm supp}(v)|.$ For a set $S\subseteq[n],$ define $\tb{v}|_S=(v_{i_1},\ldots,v_{i_{|S|}})$, where $i_j\in S$ for $1\le j\le |S|$ and $1\le i_1<\cdots<i_{|S|}\le n$.
  \item $[n,k,d]_q$ denotes a linear code of length $n$, dimension $k$ and distance $d$ over the field $\mathbb{F}_q$. We will write $[n,k,d]$ instead of $[n,k,d]_q$ when the particular choice of the field is not important.
  \item Let $\mathcal{C}$ be an $[n,k,d]$ code and $S\subseteq[n]$, $|S|=k$. We say that $S$ is an information set if the restriction $\mathcal{C}|_{S}=\mathbb{F}_q^k$.
  \item An $[n,k,d]$ code is called Maximum Distance Separable (MDS) if $d=n-k+1$. Particularly, an $[n,k,d]$ code is MDS if and only if every subset of its $k$ coordinates is an information set. Alternatively, an $[n,k,d]$ code is MDS if and only if it corrects any collection of $(n-k)$ simultaneous erasures (see \cite{Macwilliams1977}).
  \item Let $\mathcal{C}_1$ be an $[n_1,k_1,d_1]$ code and $\mathcal{C}_2$ be an $[n_2,k_2,d_2]$ code. The tensor product $\mathcal{C}_1\otimes \mathcal{C}_2$ is an $[n_1n_2,k_1k_2,d_1d_2]$ code such that the codewords of $\mathcal{C}_1\otimes \mathcal{C}_2$ are matrices of size $n_1\times n_2$, where each column belongs to $\mathcal{C}_1$ and each row belongs to $\mathcal{C}_2$. If $U\subseteq [n_1]$ is an information set of $\mathcal{C}_1$ and $V\subseteq [n_2]$ is an information set of $\mathcal{C}_2$, then $U\times V$ is an information set of $\mathcal{C}_1\otimes \mathcal{C}_2$ (see \cite{Macwilliams1977}).
  \item Let $\mathbf{I}_n$ be the $n\times n$ identity matrix. And let $\tb{1}_n$ and $\tb{0}_{n}$ be the all-one and all-zero vectors, respectively.
\end{itemize}

\subsection{Maximal recoverability for general topologies}

Let $z_1,\ldots,z_m$ be variables over the field $\mathbb{F}_q$. Consider an $(n-k)\times n$ matrix $P=\{p_{ij}\}$ where each $p_{ij}\in\mathbb{F}_p[z_1,\ldots,z_m]$ is an affine function of the $z_i$s over $\mathbb{F}_p$:
\begin{equation}\label{parityforgeneraltopologies}
p_{ij}(z_1,\ldots,z_m)=c_{ij0}+\sum\limits_{k=1}^m c_{ijk}z_k,~~~~c_{ijk}\in\mathbb{F}_p.
\end{equation}
We refer the matrix $P$ as a topology. Fix an assignment $\{z_i=\alpha_i\}_{i=1}^{m}$, where $\alpha_i\in\mathbb{F}_q$. Viewing $P(\alpha_1,\ldots,\alpha_m)$ as a parity check matrix, then it defines a linear code which is denoted by $\mathcal{C}(\alpha_1,\ldots,\alpha_m)$. And we say code $\mathcal{C}$ instantiates $P$. A set $S\subseteq[n]$ of columns of $P$ is called \emph{potentially} \emph{independent} if there exists an assignment $\{z_i=\alpha_i\}_{i=1}^{m}$ where $\alpha_i\in\mathbb{F}_q$ such that the columns of $P(\alpha_1,\ldots,\alpha_m)$ indexed by $S$ are linearly independent.

\begin{definition}\cite{Gopalan2014}\label{defGMR}
The code $\mathcal{C}(\alpha_1,\ldots,\alpha_m)$ instantiating the topology $P$ is called maximally recoverable if every set of columns that is potentially independent in $P$ is linearly independent in $P(\alpha_1,\ldots,\alpha_m)$.
\end{definition}

Using the \emph{Sparse Zeros Lemma} (see Theorem 6.13 in \cite{LN83}), Gopalan et al. \cite{Gopalan2014} proved the following upper bound on the size of field over which the maximally recoverable codes for any topologies $P$ exist.

\begin{theorem}\cite{Gopalan2014}\label{upbGMR}
Let $P\in(\mathbb{F}_p[z_1,\ldots,z_m])^{(n-k)\times n}$ be an arbitrary topology. If $q>(n-k)\cdot{n\choose {\leq n-k}}$, then there exists an MR instantiation of $P$ over the field $\mathbb{F}_q$.
\end{theorem}

\subsection{Grid-like topologies}

Unifying and generalizing a number of topologies considered both in coding theory and practice, Gopalan et al. \cite{Gopalan2017} proposed the following family of topologies called \emph{grid-like topologies} via dual constraints.

\begin{definition}\cite{Gopalan2017}\label{GLT}
Let $m\leq n$ be integers. Consider an $m\times n$ array of symbols $\{x_{ij}\}_{i\in[m],j\in[n]}$ over the field $\mathbb{F}_q$. Let $0\leq a\leq m-1$, $0\leq b\leq n-1$, and $0\leq h\leq (m-a)(n-b)-1$. Let $T_{m\times n}(a,b,h)$ denote the topology where there are $a$ parity check equations per column, $b$ parity check equations per row, and $h$ global parity check equations that depend on all symbols. Topologies of the form $T_{m\times n}(a,b,h)$ are called grid-like topologies.

Furthermore, we say a collection of arrays $\mathcal{C}$ in $\mathbb{F}_q^{m\times n}$ to be a code that instantiates the topology $T_{m\times n}(a,b,h)$, if there exist $\{\alpha_i^{(k)}\}_{i\in[m],k\in[a]}$, $\{\beta_j^{(k)}\}_{j\in[n],k\in[b]}$ and $\{\gamma_{ij}^{(k)}\}_{i\in[m],j\in[n],k\in[h]}$ in $\mathbb{F}_q$ such that for each codeword $C=(c_{ij})_{i\in[m],j\in[n]}\in \mathcal{C}$:

1. Each column $j\in[n]$ satisfies the constraints
\begin{equation}\label{columnconstrains}
\sum_{i=1}^{m}\alpha_{i}^{(k)}c_{ij}=0,~~\forall k\in[a].
\end{equation}

2. Each row $i\in[m]$ satisfies the constraints
\begin{equation}\label{rowconstrains}
\sum_{j=1}^{n}\beta_{j}^{(k)}c_{ij}=0,~~\forall k\in[b].
\end{equation}

3. All the symbols satisfy $h$ global constraints
\begin{equation}\label{globalconstrains}
\sum_{i=1}^{m}\sum_{j=1}^{n}\gamma_{ij}^{(k)}c_{ij}=0,~~\forall k\in[h].
\end{equation}

\end{definition}

\begin{definition}\label{correctableEP}
An erasure pattern is a set $E\subseteq [m]\times [n]$ of symbols. Pattern $E$ is correctable for the topology $T_{m\times n}(a,b,h)$ if there exists a code instantiating the topology where the variables $\{x_{ij}\}_{(i,j)\in E}$ can be recovered from the parity check equations (\ref{columnconstrains}), (\ref{rowconstrains}) and (\ref{globalconstrains}).
\end{definition}

Clearly, constraints in (\ref{columnconstrains}) and (\ref{rowconstrains}) guarantee the local dependencies in each column and row respectively, and constraints in (\ref{globalconstrains}) ensure some additional recoverability. Notably, constraints (\ref{columnconstrains}) specify a code $\mathcal{C}_{col}\subseteq \mathbb{F}_q^{m}$ and constraints (\ref{rowconstrains}) specify a code $\mathcal{C}_{row}\subseteq \mathbb{F}_q^n$. If $h=0$, i.e., there are no extra global constraints for all symbols, then the code specified with the settings from Definition \ref{GLT} is exactly the tensor product code $\mathcal{C}_{col}\otimes \mathcal{C}_{row}$.

\begin{definition}\label{MRforGLT}
A code $\mathcal{C}$ that instantiates the topology $T_{m\times n}(a,b,h)$ is Maximally Recoverable (MR) if it can correct every failure pattern that is correctable for the topology.
\end{definition}

The maximally recoverability requires a code that instantiates the topology $T_{m\times n}(a,b,h)$ to have many good properties, especially the MDS property.

\begin{proposition}\cite{Gopalan2017}\label{MDSforMR}
Let $\mathcal{C}$ be an MR instantiation of the topology $T_{m\times n}(a,b,h)$. We have

1. The dimension of $\mathcal{C}$ is given by
\begin{equation}
dim~\mathcal{C}=(m-a)(n-b)-h.
\end{equation}
Moreover,
\begin{equation}
dim~\mathcal{C}_{col}=m-a~~and~~dim~\mathcal{C}_{row}=n-b.
\end{equation}

2. Let $U\subseteq [m]$, $|U|=m-a$ and $V\subseteq [n]$, $|V|=n-b$ be arbitrary. Then $\mathcal{C}|_{U\times V}$ is an
\begin{equation*}
[(m-a)(n-b),~(m-a)(n-b)-h,~h+1]
\end{equation*}
MDS code. Any subset $S\subseteq U\times V$, $|S|=(m-a)(n-b)-h$ is an information set.

3. Assume
\begin{equation}
h\leq (m-a)(n-b)-\max{\{(m-a),(n-b)\}},
\end{equation}
then the code $\mathcal{C}_{col}$ is an $[m,m-a,a+1]$ MDS code and the code $\mathcal{C}_{row}$ is an $[n,n-b,b+1]$ MDS code. Moreover, for all $j\in [n]$, $\mathcal{C}$ restricted to column $j$ is the code $\mathcal{C}_{col}$, and for all $i\in [m]$, $\mathcal{C}$ restricted to row $i$ is the code $\mathcal{C}_{row}$.
\end{proposition}

Considering the topology $T_{m\times n}(a,b,0)$, the MR code $\mathcal{C}$ that instantiates this topology can be viewed as the tensor product code $\mathcal{C}_{col}\otimes \mathcal{C}_{row}$. Based on the MDS properties for both $\mathcal{C}_{col}$ and $\mathcal{C}_{row}$, for a corresponding erasure pattern, we know that if some column has less than $a+1$ erasures or some row has less than $b+1$ erasures, we can decode it. Therefore, the erasure pattern that really matters shall have at least $a+1$ erasures in each column and at least $b+1$ erasures in each row.

\begin{definition}\label{irr}
An erasure pattern $E\subseteq [m]\times[n]$ for the topology $T_{m\times n}(a,b,0)$ is called irreducible, if for any $(i,j)\in E$, $|I(j)|=|\{i'\in [m]: (i',j)\in E\}|\geq a+1$ and $|J(i)|=|\{j'\in [n]: (i,j')\in E\}|\geq b+1$.
\end{definition}

These kinds of patterns were originally mentioned in \cite{Gopalan2017} and also appeared in \cite{SRLS2018}. While Gopalan et al. \cite{Gopalan2017} were trying to characterize the correctable erasure patterns for grid-like topologies, they considered the natural question: \emph{are irreducible patterns uncorrectable?} In order to address this question, they introduced the following notion of \emph{regularity} for erasure patterns.

\begin{definition}\cite{Gopalan2017}\label{REP}
Consider the topology $T_{m\times n}(a,b,0)$ and an erasure pattern $E$. We say that $E$ is regular if for all $U\subseteq [m]$, $|U|=u$ and $V\subseteq [n]$, $|V|=v$ we have
\begin{equation}\label{regularity}
|E\cap (U\times V)|\leq va+ub-ab.
\end{equation}
\end{definition}

By reducing the regular erasure patterns to the irreducible case, the authors proved the following equivalent condition of the correctable erasure patterns for the topology $T_{m\times n}(1,b,0)$.

\begin{theorem}\cite{Gopalan2017}\label{Gop17}
An erasure pattern $E$ is correctable for the topology $T_{m\times n}(1,b,0)$ if and only if it is regular for $T_{m\times n}(1,b,0)$.
\end{theorem}

\subsection{Independent sets in hypergraphs}

A hypergraph is a pair $(V,\mathfrak{E})$, where $V$ is a finite set and $\mathfrak{E}\subseteq 2^{V}$ is a family of subsets of $V$. The elements of $V$ are called vertices and the subsets in $\mathfrak{E}$ are called hyperedges. An independent set of a hypergraph is a set of vertices containing no hyperedges and the independence number of a hypergraph is the size of its largest independent set.

There are many results on the independence number of hypergraphs obtained through different methods (see \cite{AKPSS1982}, \cite{AKS1980}, \cite{DLR1995}, \cite{Kostochka2014}). In the following section, we will apply the general lower bound derived by Kostochka et al. \cite{Kostochka2014}. Before stating their theorem, we need a few definitions and notations. Let $H(V,\mathfrak{E})$ be a hypergraph with vertex set $V$ and hyperedge set $\mathfrak{E}$. We call $H$ a $k$-uniform hypergraph, if all the hyperedges have the same size $k$, i.e., $\mathfrak{E}\subseteq {V\choose k}$. For any vertex $v\in V$, we define the degree of $v$ to be the number of hyperedges containing $v$, denoted by $d(v)$. The maximum of the degrees of all the vertices is called the maximum degree of $H$ and denoted by $\Delta(H)$. The independence number of $H$ is denoted by $\alpha(H)$. For a set $R$ of $r$ vertices, define the $r$-degree of $R$ to be the number of hyperedges containing $R$.

\begin{theorem}\cite{Kostochka2014}\label{KMV14}
Fix $r\geq 2$. There exists $c_r>0$ such that if $H$ is an $(r+1)$-graph on $n$ vertices with maximum $r$-degree $\Delta_r<n/(\log n)^{3r^2}$, then
\begin{equation}\label{hypergraphind}
\alpha(H)\geq c_r(\frac{n}{\Delta_r}\log{\frac{n}{\Delta_r}})^{\frac{1}{r}},
\end{equation}
where $c_r>0$ and $c_r\sim r/e$ as $r\rightarrow \infty$.
\end{theorem}

\section{Pseudo-parity check matrix and Regular irreducible erasure patterns}

In this section, we shall introduce two important notions: \emph{pseudo-parity check matrix} and regular irreducible erasure patterns, which are crucial in the proofs of both upper and lower bounds.

\subsection{Pseudo-parity check matrix}

Let $\mathcal{C}$ be an $[n,k]$ linear code with a parity check matrix $\mathbf{H}\in \mathbb{F}_q^{(n-k)\times n}$, then we have the following well-known fact about $\mathbf{H}$.

\begin{fact}\label{fact1}\cite{Macwilliams1977}
Assume a subset $E\subseteq [n]$ of the coordinates of $\mathcal{C}$ are erased, then they can be recovered if and only if the parity check matrix $\mathbf{H}$ restricted to coordinates in $E$ has full rank.
\end{fact}


Take $\mathcal{C}=\mathcal{C}_{col}\otimes\mathcal{C}_{row}$ as the tensor product code that instantiates the topology $T_{m\times n}(a,b,0)$, where $\mathcal{C}_{col}$ and $\mathcal{C}_{row}$ are codes specified by (\ref{columnconstrains}) and (\ref{rowconstrains}), respectively. For simplicity, for each codeword $c\in \mathcal{C}$ write
\begin{equation*}
c=(c_{11},\ldots,c_{1n},c_{21},\ldots,c_{2n},\ldots,c_{m1},\ldots,c_{mn}),
\end{equation*}
where for each $j\in[n], (c_{1j},\ldots,c_{mj})$ is a codeword in $\mathcal{C}_{col}$ and for each $i\in[m], (c_{i1},\ldots,c_{in})$ is a codeword in $\mathcal{C}_{row}$.

Denote $\mathbf{H}_{col}$ and $\mathbf{H}_{row}$ as the parity check matrices of $\mathcal{C}_{col}$ and $\mathcal{C}_{row}$ respectively, assume
\begin{equation*}\label{pmCR1}
\mathbf{H}_{col}=\left(\begin{array}{cccc}
\alpha_1^{(1)}  &  \alpha_2^{(1)}  &\cdots  &  \alpha_m^{(1)} \\ 
\alpha_1^{(2)}  &  \alpha_2^{(2)}  &  \ldots  &  \alpha_m^{(2)} \\
\vdots  &  \vdots  & \ddots  &  \vdots \\
\alpha_1^{(a)}  &  \alpha_2^{(a)}  &  \ldots  &  \alpha_m^{(a)} \\
\end{array}
\right)~\text{and}~
\mathbf{H}_{row}=\left(\begin{array}{cccc}
\beta_1^{(1)}  &  \beta_2^{(1)}  &\cdots  &  \beta_n^{(1)} \\ 
\beta_1^{(2)}  &  \beta_2^{(2)}  &  \ldots  &  \beta_n^{(2)} \\
\vdots  &  \vdots  & \ddots  &  \vdots \\
\beta_1^{(b)}  &  \beta_2^{(b)}  &  \ldots  &  \beta_n^{(b)} \\
\end{array}
\right).
\end{equation*}
Then consider the following $(an +bm )\times mn$ matrix:
\begin{equation}\label{pmT(a,b,0)}
\mathbf{H}_{(a,b,0)}=\left(\begin{array}{cccc}
\mathbf{H}_1  &  \mathbf{H}_2  &\cdots  &  \mathbf{H}_m \\ 
\mathbf{H}_{row}  &  \mathbf{0}  &  \ldots  &  \mathbf{0} \\
\mathbf{0}  &  \mathbf{H}_{row}  &  \ldots  &  \mathbf{0} \\
\vdots  &  \vdots  & \ddots  &  \vdots \\
\mathbf{0}  &  \mathbf{0}  &  \ldots  &  \mathbf{H}_{row} \\
\end{array}
\right),
\end{equation}
where
\begin{equation}\label{pmCR2}
\mathbf{H}_{i}=\left(\begin{array}{ccccc}
\vec{\alpha}_i  &  \mathbf{0}   &  \mathbf{0}  &\cdots  &  \mathbf{0} \\ 
\mathbf{0}  &  \vec{\alpha}_i  &  \mathbf{0}   &\ldots  &  \mathbf{0} \\
\vdots  &  \vdots  &  \vdots  &\ddots  &  \vdots \\
\mathbf{0}  &  \mathbf{0}  &   \mathbf{0}  &  \ldots  &  \vec{\alpha}_i \\
\end{array}
\right)~\text{is a matrix of size $(a\cdot n)\times n$}~\text{with}~
\vec{\alpha}_i=\left(\begin{array}{c}
\alpha_i^{(1)}\\ 
\alpha_i^{(2)}\\
\vdots\\
\alpha_i^{(a)}\\
\end{array}
\right).
\end{equation}
From the above construction, we can see that $\mathbf{H}_{(a,b,0)}$ includes all the parity check constraints of $\mathcal{C}$, and it can be easily verified that $\mathbf{H}_{(a,b,0)}\cdot c^{T}=0$ for each codeword $c\in \mathcal{C}$. Since the size of $\mathbf{H}_{(a,b,0)}$ is $(an +bm )\times mn$, instead of the parity check matrix of $\mathcal{C}$, it can only be regarded as an approximation of the parity check matrix. Therefore, we call $\mathbf{H}_{(a,b,0)}$ a \emph{pseudo-parity check matrix} of the code $\mathcal{C}$.

Similar to \ref{fact1}, using basic linear algebra arguments, we have the following proposition for \emph{pseudo-parity check matrix} of code $\mathcal{C}$.
\begin{proposition}\label{prop1}
Assume a subset $E\subseteq [mn]$ of the coordinates of $\mathcal{C}$ are erased, then they can be recovered if and only if the \emph{pseudo-parity check matrix} $\mathbf{H}_{(a,b,0)}$ restricted to coordinates in $E$ has full column rank.
\end{proposition}

When $a=1$, if $\mathcal{C}$ is MR, from the MDS property of the code $\mathcal{C}_{col}$, we know that $\mathbf{H}_{col}$ has rank 1. Especially, when considering the existence of MR codes for topologies $T_{m\times n}(1,b,0)$, w.l.o.g, we can fix $\mathcal{C}_{col}$ to be the simple parity code $\mathcal{P}_m$, i.e., $\mathbf{H}_{col}=(1~1~\cdots~1)$. Hence, the \emph{pseudo-parity check matrix} $\mathbf{H}_{(1,b,0)}$ of $\mathcal{C}=\mathcal{P}_m\otimes \mathcal{C}_{row}$ has the form:
\begin{equation}\label{pmT(1,b,0)}
\mathbf{H}_{(1,b,0)}=\left(\begin{array}{cccc}
\mathbf{I}_n  &  \mathbf{I}_n  &\cdots  &  \mathbf{I}_n \\[1mm]
\hdashline
\mathbf{H}_{row}  &  \mathbf{0}  &  \ldots  &  \mathbf{0} \\[1mm]
\mathbf{0}  &  \mathbf{H}_{row}  &  \ldots  &  \mathbf{0} \\[1mm]
\vdots  &  \vdots  & \ddots  &  \vdots \\[1mm]
\mathbf{0}  &  \mathbf{0}  &  \ldots  &  \mathbf{H}_{row} \\[1mm]
\end{array}
\right)=
\left(\begin{array}{c}
\mathbf{H}_1\\[1mm]
\mathbf{H}_2\\[1mm]
\end{array}
\right).
\end{equation}

\begin{remark}\label{pcMRLRC}
Let $r|n$ and $g=\frac{n}{r}$, an $(n,r,h,a,q)$-MR LRC (for specific definition, see \cite{Gopi2018}) can be viewed as an MR code for topology $T_{g\times r}(a,0,h)$. Therefore, it has simpler erasure patterns compared to the tensor product cases. And instead of using the \emph{pseudo-parity check matrix}, it can be verified that the parity check matrix of any $(n,r,h,a,q)$-MR LRC admits the form
\begin{equation*}
\mathbf{H}=\left(\begin{array}{cccc}
\mathbf{A}_1 &  \mathbf{A}_2  &\cdots  &  \mathbf{A}_{g} \\[1mm]
\mathbf{H}_{1}  &  \mathbf{0}  &  \ldots  &  \mathbf{0} \\[1mm]
\mathbf{0}  &  \mathbf{H}_{2}  &  \ldots  &  \mathbf{0} \\[1mm]
\vdots  &  \vdots  & \ddots  &  \vdots \\[1mm]
\mathbf{0}  &  \mathbf{0}  &  \ldots  &  \mathbf{H}_{g} \\[1mm]
\end{array}
\right),
\end{equation*}
where for each $i\in[g]$, $\mathbf{H}_i$ is a parity check matrix of an $[r,r-a,a+1]$ MDS code and $\mathbf{A}_i$ is an $h\times r$ matrix over $\mathbb{F}_q$ corresponding to the global parities.

Compared to MR LRCs, MR codes for topologies $T_{m\times n}(a,b,0)$ have another difference. For an $(n,r,h,a,q)$-MR LRC, the $[r,r-a,a+1]$ MDS codes within each local group can be different, this results in that the corresponding parity check matrix $\mathbf{H}$ above can admit different $\mathbf{H}_i$s. However, since an MR code for topology $T_{m\times n}(a,b,0)$ is actually a tensor product code $\mathcal{C}=\mathcal{C}_{col}\otimes\mathcal{C}_{row}$. Thus, for each $i\in[m]$, if we take coordinates in $\{n(i-1)+1,\ldots,ni\}$ as a local group, once the code $\mathcal{C}_{row}$ is fixed, the corresponding $[n,n-b,b+1]$ MDS codes within each local group are all $\mathcal{C}_{row}$ and the corresponding parity check matrices in $\mathbf{H}_{(a,b,0)}$ are all $\mathbf{H}_{row}$.
\end{remark}

\subsection{Regular irreducible erasure patterns}

Let $E\in[m]\times[n]$ be an erasure pattern of the topology $T_{m\times n}(a,b,0)$, then it can be presented in the following form:
\begin{equation*}\label{erasurepattern}
E=\left(\begin{array}{ccccccc}
*  &  *  &  *  &  *  &  \cdots  &  \circ  &  \circ\\[1mm]
*  &  *  &  \circ  &  *  &  \cdots  &  \circ  &  *\\[1mm]
\circ  &  *  &  \circ  &  *  &  \cdots  &  *  &  \circ\\[1mm]
\vdots  &  \vdots  &  \vdots  &  \vdots  &  \ddots  &  \vdots  &  \vdots\\[1mm]
\circ  &  \circ  &  \circ  &  *  &  \cdots  &  *  &  *\\[1mm]
\end{array}
\right),
\end{equation*}
where $*$ stands for the erasure and $\circ$ stands for the non-erasure. Give two different erasure patterns $E_1$ and $E_2$, we say that $E_1$ and $E_2$ are of the same type, if $E_2$ can be obtained from $E_1$ by applying elementary row and column transformations. 

For a reducible erasure pattern $E$, there exists some $i_0\in[m]$ or $j_0\in[n]$, such that the number of the erasures in $E\cap [i_0]\times [n]$ or $E\cap [m]\times [j_0]$ is less than $b+1$ or $a+1$. Therefore, from the MDS properties of the code $\mathcal{C}_{row}$ and $\mathcal{C}_{col}$, erasures in $E\cap [i_0]\times [n]$ or $E\cap [m]\times [j_0]$ can be simply repaired by using only the parities within $\mathbf{H}_{row}$ or $\mathbf{H}_{col}$. Hence, the very erasure patterns that affect the MR property of the code $\mathcal{C}$ are irreducible erasure patterns. In other words, if we can construct a code $\mathcal{C}$ instantiating the topology $T_{m\times n}(a,b,0)$ that can correct all correctable irreducible erasure patterns, then this code $\mathcal{C}$ is an MR instantiation for the topology $T_{m\times n}(a,b,0)$.

Now, we focus on the irreducible erasure patterns that are correctable. Given an irreducible erasure pattern $E$, denote $|E|$ as the number of $*$s in $E$, $U_E=\{i\in[m]: \exists j\in[n]~\text{such~that} ~E(i,j)=*\}$ and $V_E=\{j\in[n]: \exists i\in[m]~\text{such that}~ E(i,j)=*\}$. From the irreducibility of $E$, we have
\begin{equation*}
|E|\geq (a+1)|V_E|~and~|E|\geq (b+1)|U_E|.
\end{equation*}
Meanwhile, from Theorem~\ref{Gop17}, we know that for topology $T_{m\times n}(1,b,0)$, an erasure pattern $E$ is correctable if and only if $E$ is regular. Thus we have
\begin{equation*}
|E|=|E\cap (U_E\times V_E)|\leq a|V_E|+b|U_E|-ab=|V_E|+b|U_E|-b.
\end{equation*}
Combining the above three inequalities together, we have
\begin{equation}\label{ineq0}
|U_E|+b\leq|V_E|\leq b|U_E|-b,
\end{equation}
for every correctable irreducible erasure patterns $E$ in $T_{m\times n}(1,b,0)$. Therefore,
\begin{equation}\label{ineq1}
\max\{{2(|U_E|+b),(b+1)|U_E|}\}\leq |E|\leq 2b(|U_E|-1),
\end{equation}
which indicates that once $|U_E|$ (or $|V_E|$) is given, the magnitude of $|E|$ can not be too large.

Denote $\mathcal{E}$ as the set of all the types of regular irreducible erasure patterns for topology $T_{m\times n}(1,b,0)$, i.e., for each $E\in\mathcal{E}$, one can regard $E$ as a representative of all the erasure patterns that have the same type as $E$. Since $U_E\subseteq[m]$ and $V_E\subseteq [n]$, for each regular irreducible erasure pattern $E$, we have $|V_E|\leq b(m-1)$ and $|E|\leq 2b(m-1)$. For convenience, we can take each type of erasure patterns in $\mathcal{E}$ as a submatrix of an $m\times b(m-1)$ matrix with elements from $\{*,\circ\}$. Therefore, we can obtain the following upper bound of $|\mathcal{E}|$:
\begin{equation}\label{typenumber}
|\mathcal{E}|\leq {m\cdot b(m-1)\choose \leq 2b(m-1)}.
\end{equation}

\section{A polynomial upper bound on the maximum field size required for MR codes}\label{secup}

In this section, we take the prime $p=2$, which is the natural setting for distributed storage. And we will establish our polynomial upper bound on the minimal field size required for MR codes that instantiate the topology $T_{m\times n}(1,b,0)$.

\begin{theorem}\label{upbT(1,b,0)}
Let $m,b\geq 1$. Then for any $q\geq C_0(m,b)\cdot n^{2b(m-1)}+n^{(b-1)}$, there exists an MR code $\mathcal{C}$ that instantiates the topology $T_{m\times n}(1,b,0)$ over the field $\mathbb{F}_q$, where $C_0(m,b)=(m+1)!\cdot{m\cdot b(m-1)\choose \leq2b(m-1)}$.
\end{theorem}

In order to do this, we will exhibit a column code $\mathcal{C}_{col}$ and a row code $\mathcal{C}_{row}$ over a relative small field, so that for every correctable irreducible erasure pattern $E$, the code $\mathcal{C}_{col}\otimes\mathcal{C}_{row}$ can correct $E$. Thus the tensor product code $\mathcal{C}=\mathcal{C}_{col}\otimes\mathcal{C}_{row}$ is an MR code that instantiates the topology $T_{m\times n}(1,b,0)$. We also need the following lemma known as the Combinatorial Nullstellensatz.

\begin{lemma}\label{CN}(Combinatorial Nullstellensatz) \cite{Alon1999}
Let $\mathbb{F}$ be an arbitrary field, let $P\in\mathbb{F}[t_1,\ldots,t_n]$ be a polynomial of degree $d$ which contains a non-zero coefficient at $t_1^{d_1}\cdots t_n^{d_n}$ with $d_1+\cdots+d_n=d$, and let $S_1,\ldots,S_n$ be subsets of $\mathbb{F}$ such that $|S_i|>d_i$ for all $1\leq i\leq n$. Then there exist $x_1\in S_1,\ldots,x_n\in S_n$ such that $P(x_1,\ldots,x_n)\neq 0$.
\end{lemma}

\begin{proof}[Proof of Theorem \ref{upbT(1,b,0)}]
Since the case $m=1$ is trivial, w.l.o.g., we assume $m\geq 2$. For simplicity, we fix $\mathcal{C}_{col}$ as the simple parity code $\mathcal{P}_m$ and focus on obtaining the code $\mathcal{C}_{row}$.

Denote $\mathcal{E}$ as the set of all the types of regular irreducible erasure patterns for topology $T_{m\times n}(1,b,0)$. Assume the parity check matrix of the code $\mathcal{C}_{row}$ is $\mathbf{H}_{row}$, then the \emph{pseudo-parity check matrix} $\mathbf{H}$ is of the form in (\ref{pmT(1,b,0)}). Thus, our goal is to construct a $b\times n$ matrix $\mathbf{H}_{row}$ such that:
\begin{enumerate}
  \item[(i)] Every $b$ distinct columns of $\mathbf{H}_{row}$ are linearly independent.
  \item[(ii)] For each regular irreducible erasure pattern $E$, the \emph{pseudo-parity check matrix} $\mathbf{H}\in\mathbb{F}_q^{(n+bm)\times mn}$ of $\mathcal{C}$
  satisfies: $rank(\mathbf{H}|_E)=|E|$.
\end{enumerate}

Given a regular irreducible erasure pattern $E\in[m]\times [n]$, w.l.o.g., assume $U_E=[u_0]\subseteq[m]$ and $V_E=[v_0]\subseteq[n]$, then $E$ has the form
\begin{equation*}\label{erasurepattern1}
E=\left(\begin{array}{ccccccc}
*  &  *  &  *  &  *  &  \cdots  &  \circ  &  \circ\\[1mm]
*  &  *  &  *  &  \circ  &  \cdots  &  *  &  \circ\\[1mm]
\circ  &  *  &  \circ  &  *  &  \cdots  &  *  &  \circ\\[1mm]
\vdots  &  \vdots  &  \vdots  &  \vdots  &  \ddots  &  \vdots  &  \vdots\\[1mm]
\circ  &  \circ  &  \circ  &  *  &  \cdots  &  *  &  *\\[1mm]
\end{array}
\right)
=\left(\begin{array}{c}
E_{1}\\ 
E_{2}\\
E_{3}\\
\vdots\\
E_{u_0}\\
\end{array}
\right),
\end{equation*}
where $E_i$ represents the sub-erasure pattern of $E$ over the $i_{th}$ row. Thus
\begin{equation*}
\mathbf{H}|_E=\left(\begin{array}{cccc}
\mathbf{I}_n|_{E_{1}}  &  \mathbf{I}_n|_{E_{2}}  &\cdots  &  \mathbf{I}_n|_{E_{u_0}} \\[1mm]
\mathbf{H}_{row}|_{E_{1}}  &  \mathbf{0}  &  \ldots  &  \mathbf{0} \\[1mm]
\mathbf{0}  &  \mathbf{H}_{row}|_{E_{2}}  &  \ldots  &  \mathbf{0} \\[1mm]
\vdots  &  \vdots  & \ddots  &  \vdots \\[1mm]
\mathbf{0}  &  \mathbf{0}  &  \ldots  &  \mathbf{H}_{row}|_{E_{u_0}} \\[1mm]
\end{array}
\right)=
\left(\begin{array}{c}
\mathbf{H}_1|_E\\[1mm]
\mathbf{H}_2|_E\\[1mm]
\end{array}
\right).
\end{equation*}
Let $supp(E_i)=\{j\in[n]:E_i(j)=*\}$. Since $\bigcup_{i=1}^{u_0} supp(E_i)=[v_0]$, by applying elementary row and column transformations, we have
\begin{equation}\label{simpleppcm}
\mathbf{H}|_E=\left(\begin{array}{cc}
\mathbf{I}_{v_0}  &  \mathbf{0}_{v_0\times(|E|-v_0)}\\[1mm]
\mathbf{A}_{u_0b\times v_0}  &  \mathbf{B}_{u_0b\times(|E|-v_0)}\\[1mm]
\mathbf{0}_{(n-v_0)\times(v_0)}  &  \mathbf{0}_{(n-v_0)\times(|E|-v_0)}\\[1mm]
\end{array}
\right),
\end{equation}
where $\mathbf{A}$ consists of all the columns in $\mathbf{H}_2|_E$ corresponding to an $\mathbf{I}_{v_0}$ in $\mathbf{H}_1|_E$ and $\mathbf{B}$ consists of all the rest columns in $\mathbf{H}_2|_E$ by substituting columns of $\mathbf{A}$ with the same parts in $\mathbf{H}_1|_E$. Thus a non-zero element of $\mathbf{A}$ equals to some $h_{ij}$ in $\mathbf{H}_{row}$ and a non-zero element of $\mathbf{B}$ equals to $h_{ij}$ or $-h_{ij}$ for some $h_{ij}$ in $\mathbf{H}_{row}$. For example, take $U_E=V_E=\{1,2,3\}$, $\mathbf{H}_{row}=(\tb{h}_1,\ldots,\tb{h}_n)$ and
\begin{equation*}\label{eg}
E=\left(\begin{array}{ccc}
*  &  *  &  \circ\\[1mm]
*  &  \circ  &  *\\[1mm]
\circ  &  \circ  &  *\\[1mm]
\end{array}
\right),
\end{equation*}
then
\begin{equation*}
\mathbf{H}|_E
=\left(\begin{array}{ccccc}
1 & & & & \\
& 1 & & & \\
& & 1 & & \\
\tb{h}_1 & \tb{h}_2 & & -\tb{h}_1 & \\
& & \tb{h}_3 & \tb{h}_1 & -\tb{h}_3 \\
& & & & \tb{h}_3\\
\multicolumn{3}{c}{\raisebox{1ex}[0pt]{$\mathbf{0}_{(n-3)\times 3}$}} & \multicolumn{2}{c}{\raisebox{1ex}[0pt]{$\mathbf{0}_{(n-3)\times 2}$}}\\
\end{array}
\right).
\end{equation*}

From the above simplified form of $\mathbf{H}|_E$ in (\ref{simpleppcm}), we have
\begin{align*}
rank(\mathbf{H}|_E)&=rank(\mathbf{I}_{v_0})+rank(\mathbf{B}_{u_0b\times(|E|-v_0)}) \\
&=v_0+rank(\mathbf{B}_{u_0b\times(|E|-v_0)}).
\end{align*}

By Definition \ref{REP} and (\ref{ineq1}), we have $|E|\leq v_0+u_0b-b$. Thus $rank(\mathbf{H}|_E)=|E|$ if and only if there exists an $(|E|-v_0)\times(|E|-v_0)$ minor $\mathbf{B}'$ in $\mathbf{B}$ such that $det(\mathbf{B}')\neq 0$.

Now, take
\begin{equation*}
\mathbf{H}_{row}
=\left(\begin{array}{ccccc}
x_{11}  &  x_{12}  &  x_{13}  &  \cdots  &  x_{1n}\\[1mm]
x_{21}  &  x_{22}  &  x_{23}  &  \cdots  &  x_{2n}\\[1mm]
x_{31}  &  x_{32}  &  x_{33}  &  \cdots  &  x_{3n}\\[1mm]
\vdots  &  \vdots  &  \vdots  &  \ddots  &  \vdots\\[1mm]
x_{b1}  &  x_{b2}  &  x_{b3}  &  \cdots  &  x_{bn}\\[1mm]
\end{array}
\right),
\end{equation*}
where each $x_{ij}$ is a variable over $\mathbb{F}_q$. Therefore, our goal is to find a proper valuation of these $x_{ij}'s$ over $\mathbb{F}_q$ such that the resulting matrix $\mathbf{H}_{row}$ satisfies both requirement (i) and requirement (ii).

\begin{itemize}
  \item For requirement (i)
\end{itemize}

For any $J=\{j_1,\ldots,j_b\}\subseteq [n]$, let $\mathbf{M}_{J}$ be the $b\times b$ submatrix of $\mathbf{H}_{row}$ formed by the $b$ columns indicated by $J$, i.e.,
\begin{equation*}
\mathbf{M}_{J}
=\left(\begin{array}{ccccc}
x_{1j_1}  &  x_{1j_2}  &  x_{1j_3}  &  \cdots  &  x_{1j_b}\\[1mm]
x_{2j_1}  &  x_{2j_2}  &  x_{2j_3}  &  \cdots  &  x_{2j_b}\\[1mm]
x_{3j_1}  &  x_{3j_2}  &  x_{3j_3}  &  \cdots  &  x_{3j_b}\\[1mm]
\vdots  &  \vdots  &  \vdots  &  \ddots  &  \vdots\\[1mm]
x_{bj_1}  &  x_{bj_2}  &  x_{bj_3}  &  \cdots  &  x_{bj_b}\\[1mm]
\end{array}
\right).
\end{equation*}
Define
\begin{equation*}
P=\prod_{J\in{[n]\choose b}}det(\mathbf{M}_{J}).
\end{equation*}
Since each $det(\mathbf{M}_{J})$ is a homogeneous polynomial of degree $b$, we know that $P$ is a homogeneous polynomial of degree $b{n\choose b}$, and each variable $x_{ij}$ has degree at most ${{n-1}\choose {b-1}}$. According to the definition of $P$, if there is a valuation $(h_{11},\ldots,h_{bn})$ of $(x_{11},\ldots,x_{bn})$ such that $P(h_{11},\ldots,h_{bn})\neq 0$, then the resulting matrix $\mathbf{H}_{row}=(h_{ij})_{i\in[b],j\in[n]}$ satisfies requirement (i).

\begin{itemize}
  \item For requirement (ii)
\end{itemize}

For each regular irreducible erasure pattern $E\in[m]\times [n]$, set $|U_E|=u_0$ and $|V_E|=v_0$ and consider the $u_0b\times(|E|-v_0)$ submatrix $\mathbf{B}(E)$ of $\mathbf{H}|_E$ in (\ref{simpleppcm}). For each $(|E|-v_0)\times(|E|-v_0)$ minor $\mathbf{B}'(E)$ in $\mathbf{B}(E)$, $det(\mathbf{B}'(E))$ can be viewed as a multi-variable polynomial in $\mathbb{F}_q[x_{11},\ldots,x_{bn}]$ with degree at most $|E|-v_0$. Since each non-zero element of $\mathbf{B}$ equals to $x_{ij}$ or $-x_{ij}$ for some $x_{ij}$ in $\mathbf{H}_{row}$, and each variable $x_{ij}$ appears in at most $u_0-1$ columns of $\mathbf{B}$, thus for each minor $\mathbf{B}'(E)$ we have
\begin{equation}\label{det}
det(\mathbf{B}'(E))=\sum_{\substack{\sum_{1\leq i\leq b,1\leq j\leq n} a_{ij}=|E|-v_0,\\ 0\leq a_{ij}\leq u_0-1}} c_{(a_{11},\ldots,a_{bn})}\cdot x_{11}^{a_{11}}x_{12}^{a_{12}}\cdots x_{bn}^{a_{bn}},
\end{equation}
where $c_{(a_{11},\ldots,a_{bn})}$ equals to $0,1$ or $-1$.

Noticed that the structure of $\mathbf{B}'(E)$ is determined by the erasure pattern. Therefore, once $E$ is given, for each minor $\mathbf{B}'(E)$ in $\mathbf{B}(E)$, $det(\mathbf{B}'(E))$ can be viewed as a polynomial in $\mathbb{F}_2[x_{11},\ldots,x_{bn}]$ with a fixed form.

Since for each regular irreducible erasure pattern $E$, $|E|\leq v_0+bm-b$. Thus, when $q >bm-b\geq \deg(det(\mathbf{B}'(E)))$, $det(\mathbf{B}'(E))|_{\mathbb{F}_q^{(bn)}}\equiv 0$ if and only if $det(\mathbf{B}'(E))=\mathbf{0}$ (i.e. the zero polynomial).

According to the proof of Theorem \ref{Gop17} in \cite{Gopalan2017}, when the size of the field is large enough, there exists a code $\mathcal{C}_{0}$ such that the tensor product code $\mathcal{C}=\mathcal{P}_m\otimes \mathcal{C}_{0}$ can correct $E$. This means that there exists a valuation of the $bn$ variables in $\mathbf{H}_{row}$ such that $det(\mathbf{B}'(E))\neq 0$ for some $(|E|-|V_E|)\times(|E|-|V_E|)$ minor $\mathbf{B}'(E)$ in $\mathbf{B}(E)$. By this, we know that the multi-variable polynomial $det(\mathbf{B}'(E))$ corresponding to this minor $\mathbf{B}'(E)$ can not be zero polynomial. From the previous analysis, we know that the form of this polynomial $det(\mathbf{B}'(E))$ is irrelevant to the size of the field. Therefore, for any $q> bm-b$ as a power of $2$, this $det(\mathbf{B}'(E))$ is a non-zero polynomial in $\mathbb{F}_q[x_{11},\ldots,x_{bn}]$.

For each regular irreducible erasure pattern $E$, denote $f_E$ as the non-zero determinant polynomial corresponding to some $(|E|-|V_E|)\times(|E|-|V_E|)$ minor $\mathbf{B}'(E)$ in $\mathbf{B}(E)$. Define
\begin{equation}\label{poly1}
F=\prod_{\substack{E\in[m]\times [n], \\E\text{~is a regular irreducible erasure pattern}}}f_E.
\end{equation}
Similarly, if there is a valuation $(h_{11},\ldots,h_{bn})$ of $(x_{11},\ldots,x_{bn})$ such that $F(h_{11},\ldots,h_{bn})\neq 0$, then the resulting matrix $\mathbf{H}_{row}=(h_{ij})_{i\in[b],j\in[n]}$ satisfies requirement (ii).

In order to apply the Combinatorial Nullstellensatz, we shall estimate the degree of each variable in $F$. Noted that
\begin{equation*}
F=\prod_{E^{*}\in\mathcal{E}}\prod_{\substack{E \text{~is a regular irreducible}\\ \text{~erasure pattern of the same type with~} E^{*}}}f_E,
\end{equation*}
and for each $E^{*}=(E_1,E_2,\ldots,E_m)^{T}\in \mathcal{E}$, there are at most $m!\cdot\prod_{i=1}^{m}{n\choose |E_i|}$ different regular irreducible erasure pattern of the same type with $E^{*}$. By (\ref{det}), for every regular irreducible erasure pattern $E$, we have the degree of each variable $x_{ij}$ in $f_{E}$ is at most $m-1$. Therefore, the degree of each variable $x_{ij}$ in $F$ is at most $(m-1)\cdot |\mathcal{E}|\cdot m!\cdot\prod_{i=1}^{m}{n\choose |E_i|}$. Since for each regular irreducible erasure pattern $E^{*}\in\mathcal{E}$, $\sum_{i=1}^{m}|E_i|=|E^{*}|\leq 2b(m-1)$, combined with the inequality (\ref{typenumber}), we have
\begin{equation*}
(m-1)\cdot|\mathcal{E}|\cdot m!\cdot\prod_{i=1}^{m}{n\choose |E_i|}\leq {m\cdot b(m-1)\choose \leq 2b(m-1)}\cdot (m+1)!\cdot n^{2b(m-1)}.
\end{equation*}

Now, consider the polynomial $P\cdot F$, by Lemma \ref{CN}, there is a valuation $(h_{11},\ldots,h_{bn})$ of $(x_{11},\ldots,x_{bn})$ over a field $\mathbb{F}_q$ of size \begin{align*}
q&={m\cdot b(m-1)\choose \leq 2b(m-1)}\cdot (m+1)!\cdot n^{2b(m-1)}+n^{(b-1)}\\
&>{{n-1}\choose {b-1}}+(m-1)\cdot |\mathcal{E}|\cdot m!\cdot\prod_{i=1}^{m}{n\choose |E_i|},
\end{align*}
such that $P\cdot F(h_{11},\ldots,h_{bn})\neq 0$. Therefore, the corresponding matrix $\mathbf{H}_{row}=(h_{ij})_{i\in[b],j\in[n]}$ is the objective matrix satisfying both requirement (i) and requirement (ii). This completes the proof.

\end{proof}

\begin{remark}
Considering the MR codes for topologies $T_{m\times n}(1,b,0)$, the general bound given by Gopalan et al. \cite{Gopalan2014} is
\begin{equation}\label{poly2}
q>(n+bm-b)\cdot{mn\choose {\leq n+bm-b}}=\Omega((n+bm-b)^2(\frac{mn}{n+bm-b})^{(n+bm-b)}),
\end{equation}
which is exponentially increasing for both $m$ and $n$, while the bound given by Theorem \ref{upbT(1,b,0)} is only a polynomial of $n$.

But, even so, when considering the growth rate corresponding to $m$,
\begin{equation*}
q={m\cdot b(m-1)\choose \leq 2b(m-1)}\cdot (m+1)!\cdot n^{2b(m-1)}+n^{(b-1)}=\Omega(m^{2b(m-1)+m}n^{2b(m-1)})
\end{equation*}
grows exponentially.

Actually, $m$ is often considered as the number of data centers in practice, which is very small compared to $n$. Therefore, the when $n\gg m$, the bound given by Theorem \ref{upbT(1,b,0)} is better than that in \cite{Gopalan2014}.
\end{remark}

\section{MR codes for topologies $T_{4\times n}(1,2,0)$ and $T_{3\times n}(1,3,0)$}

In this section, we will discuss the MR codes that instantiate topologies $T_{4\times n}(1,2,0)$ and $T_{3\times n}(1,3,0)$. For each topology, we will prove a non-trivial lower bound and an improved upper bound on the field size required for the existence of corresponding MR codes.

\subsection{MR codes for topologies $T_{4\times n}(1,2,0)$}

First, using the results from Section III.B, we will give a complete characterization of the regular irreducible erasure patterns for topology $T_{4\times n}(1,2,0)$.

Denote $\mathcal{E}$ as the set of all the types of regular irreducible erasure patterns for topology $T_{4\times n}(1,2,0)$. For each $E\in\mathcal{E}$, by (\ref{ineq0}), we have $|U_E|+2\leq|V_E|\leq 2|U_E|-2$, which leads to $|U_E|\geq 4$. Since $U_E\subseteq [m]=[4]$, we have $|U_E|=4$ and $|V_E|=6$. Therefore, from (\ref{ineq1}), we have $|E|=12$ and from the irreducibility, each erasure pattern has exactly $2$ erasures in each column and $3$ erasures in each row. Finally by checking the regularity case by case, there are $2$ different types of erasure patterns in $\mathcal{E}$:
\begin{itemize}
  \item Type I
  \begin{equation*}\label{type1}
  E_1=\left(\begin{array}{cccccc}
  *  &  *  &  *  &  \circ  &  \circ  &  \circ  \\[1mm]
  *  &  *  &  \circ  &  *  &  \circ  &  \circ  \\[1mm]
  \circ  &  \circ  &  *  &  \circ  &  *  &  *  \\[1mm]
  \circ  &  \circ  &  \circ  &  *  &  *  &  *  \\[1mm]
  \end{array}
  \right),
  \end{equation*}
  \item Type II
  \begin{equation*}\label{type2}
  E_2=\left(\begin{array}{cccccc}
  *  &  *  &  *  &  \circ  &  \circ  &  \circ  \\[1mm]
  *  &  \circ  &  \circ  &  *  &  *  &  \circ  \\[1mm]
  \circ  &  *  &  \circ  &  *  &  \circ  &  *  \\[1mm]
  \circ  &  \circ  &  *  &  \circ  &  *  &  *  \\[1mm]
  \end{array}
  \right).
  \end{equation*}
\end{itemize}

\subsubsection{Upper bound}~

Now, we are going to prove the following existence result for MR codes instantiating the topology $T_{4\times n}(1,2,0)$, which improves the general upper bound from Theorem~\ref{upbT(1,b,0)} for this special topology.

\begin{theorem}\label{upbT(1,4,0)}
For any $q>\frac{n^{5}}{\log (n)}\cdot C_1$, there exists an MR code $\mathcal{C}$ that instantiates the topology $T_{4\times n}(1,2,0)$ over the field $\mathbb{F}_q$, where $C_1\geq (\frac{10}{c_5})^5$ is an absolute constant.
\end{theorem}

\begin{proof}
Similar to the proof of Theorem \ref{upbT(1,b,0)}, let $\mathcal{C}_{col}$ be the simple parity code $\mathcal{P}_4$. Our goal is to construct a $2\times n$ matrix $\mathbf{H}_{row}$ such that:
\begin{enumerate}
  \item[(i)] Every $2$ distinct columns of $\mathbf{H}_{row}$ are linearly independent.
  \item[(ii)] For each regular irreducible erasure pattern $E$ of Type I or Type II, the \emph{pseudo-parity check matrix} $\mathbf{H}\in\mathbb{F}_q^{(n+8)\times 4n}$ of $\mathcal{P}_4\otimes\mathcal{C}_{row}$ satisfies: $rank(\mathbf{H}|_{E})=12$.
\end{enumerate}
Different from the general strategy, we are going to obtain an objective matrix based on the Vandermonde matrix.

Suppose there exists an objective matrix $\mathbf{A}_{0}$ of the form
\begin{equation*}
\mathbf{A}_{0}=\left(\begin{array}{ccccc}
1 & 1 & 1 & \cdots & 1 \\
a_1 & a_2 & a_3 & \cdots & a_n \\
\end{array}
\right),
\end{equation*}
where $\{a_i\}_{i\in[n]}$ are pairwise distinct elements in $\mathbb{F}_q$. Then the distinctness of $\{a_i\}_{i\in[n]}$ guarantees that $\mathbf{A}_{0}$ satisfies (i).

Now take $\mathbf{H}_{row}=\mathbf{A}_0$ and consider the \emph{pseudo-parity check matrix} $\mathbf{H_{A_0}}$. For each $s\in[2]$, we have
\begin{equation*}
\mathbf{H_{A_0}}|_{E_s}=\left(\begin{array}{cc}
\mathbf{I}_{6}  &  \mathbf{0}_{6 \times 6}\\[1mm]
\mathbf{A}^{(s)}_{8 \times 6}  &  \mathbf{B}^{(s)}_{8\times 6}\\[1mm]
\mathbf{0}_{(n-6)\times 6}  &  \mathbf{0}_{(n-6)\times 6}\\[1mm]
\end{array}
\right),
\end{equation*}
where
\begin{equation*}
\mathbf{A}^{(1)}=
\left(\begin{array}{cccccc}
1 & 1 & 1 & & & \\
a_1 & a_2 & a_3 & & & \\
& & & 1 & & \\
& & & a_4& & \\
& & & & 1 & 1 \\
& & & & a_5 & a_6 \\
& & & & & \\
& \multicolumn{4}{c}{\raisebox{1ex}[0pt]{$\mathbf{0}_{2\times 6}$}} &\\
\end{array}
\right)~\text{and}~
\mathbf{A}^{(2)}=
\left(\begin{array}{cccccc}
1 & 1 & 1 & & & \\
a_1 & a_2 & a_3 & & & \\
& & & 1 & 1 & \\
& & & a_4 & a_5 & \\
& & & & & 1 \\
& & & & & a_6 \\
& & & & & \\
& \multicolumn{4}{c}{\raisebox{1ex}[0pt]{$\mathbf{0}_{2\times 6}$}} &\\
\end{array}
\right),
\end{equation*}
\begin{equation*}
\mathbf{B}^{(1)}=
\left(\begin{array}{cccccc}
-1 & -1 & -1 & & & \\
-a_1 & -a_2 & -a_3 & & & \\
1 & 1 & & -1 & & \\
a_1& a_2& & -a_4& & \\
& & 1 & & -1 & -1\\
& & a_3 & & -a_5 & -a_6\\
& & & 1 & 1 & 1\\
& & & a_4 & a_5 & a_6\\
\end{array}
\right)~\text{and}~
\mathbf{B}^{(2)}=
\left(\begin{array}{cccccc}
-1 & -1 & -1 & & & \\
-a_1 & -a_2 & -a_3 & & & \\
1 & & & -1 & -1 & \\
a_1& & & -a_4 & -a_5 & \\
& 1 & & 1 & & -1\\
& a_2 & & a_4 & & -a_6\\
& & 1 & & 1 & 1\\
& & a_3 & & a_5 & a_6\\
\end{array}
\right).
\end{equation*}
Since $\mathbf{B}^{(s)}$ can be simplified as
\begin{equation*}
\mathbf{B}^{(1)}=
\left(\begin{array}{cccccc}
1 & & & & & \\
& 1 & & & & \\
& & 1 & & & \\
& & & a_2-a_1 & a_4-a_1 & \\
& & & & a_3-a_4 & \\
& & & & a_5-a_4 & a_6-a_5\\
& & & & & \\
& \multicolumn{4}{c}{\raisebox{1ex}[0pt]{$\mathbf{0}_{2\times 6}$}} &\\
\end{array}
\right)~\text{and}~
\mathbf{B}^{(2)}=
\left(\begin{array}{cccccc}
1 & & & & & \\
& 1 & & & & \\
& & 1 & & & \\
& & & a_1-a_4 & a_1-a_5 & \\
& & & a_4-a_2 & & a_2-a_6\\
& & & & a_5-a_3 & a_6-a_3\\
& & & & & \\
& \multicolumn{4}{c}{\raisebox{1ex}[0pt]{$\mathbf{0}_{2\times 6}$}} &\\
\end{array}
\right),
\end{equation*}
we have
\begin{itemize}
  \item $rank(\mathbf{B}^{(1)})=6$ if and only if $(a_2-a_1)(a_4-a_3)(a_6-a_5)\neq 0$.
  \item $rank(\mathbf{B}^{(2)})=6$ if and only if $(a_1-a_4)(a_2-a_6)(a_3-a_5)-(a_2-a_4)(a_1-a_5)(a_3-a_6)\neq 0$.
\end{itemize}
Take $f(x_1,x_2,\ldots,x_6)=(x_1-x_4)(x_2-x_6)(x_3-x_5)-(x_2-x_4)(x_1-x_5)(x_3-x_6)$, we have $deg(f)=3$. From the assumption that $\{a_i\}_{i\in[n]}$ are pairwise distinct, we know that erasure patterns of Type I can be easily corrected. Then if we also want to correct all erasure patterns of Type II, $\{a_i\}_{i\in[n]}$ only need to have the property that for any $\{a_{i_1},a_{i_2},\ldots,a_{i_6}\}\subseteq \{a_i\}_{i\in[n]}$ and each $\pi\in S_6$, $f(a_{i_{\pi(1)}},\ldots,a_{i_{\pi(6)}})\neq 0$.

Different from the proof of Theorem \ref{upbT(1,b,0)}, here we use the hypergraph independent set approach.

Let $\mathcal{H}$ be a $6$-uniform hypergraph with vertex set $\mathbb{F}_q$, each set of $6$ vertices $\{v_{1},\ldots,v_{6}\}$ forms a $6$-hyperedge if and only if $f(v_{\pi(1)},\ldots,v_{\pi(6)})=0$ for some $\pi \in S_{6}$. From the construction of the hypergraph $\mathcal{H}$, if there exists an independent set $I$ such that $|I|\geq n$, then we can construct an objective matrix $\mathbf{A}_{0}$ by arbitrarily choosing $n$ different vertices from $I$ as elements for its $2_{nd}$ row.

Since $deg_{x_i}(f)=1$ for each $x_i$, and $f(v_1,\ldots,x_i,\ldots,v_6)$ is a non-zero polynomial for any 5-subset $\{v_{j}\}_{j\in[6]\setminus\{i\}}\subseteq \mathbb{F}_q$. Thus the maximal $5$-degree of $\mathcal{H}$ $\Delta_{5}(\mathcal{H})\leq 6!$. By Theorem \ref{KMV14},
\begin{equation*}
\alpha(\mathcal{H})\geq \frac{c_5}{5}(q\log{q})^{\frac{1}{5}}>n.
\end{equation*}
Denote $I(\mathcal{H})$ as the maximum independent set in $\mathcal{H}$, therefore, there exists a subset $A=\{a_1\ldots,a_n\}\subseteq \mathbb{F}_q$ such that the matrix $\mathbf{A}_0$ of the form
\begin{equation*}
\mathbf{A}_{0}=\left(\begin{array}{ccccc}
1 & 1 & 1 & \cdots & 1 \\
a_1 & a_2 & a_3 & \cdots & a_n \\
\end{array}
\right)
\end{equation*}
satisfies both (i) and (ii). Thus, the resulting tensor product code $\mathcal{C}=\mathcal{P}_{4}\otimes\mathcal{C}_{row}$ is an MR code instantiating topology $T_{4\times n}(1,2,0)$.


\end{proof}

\subsubsection{Lower bound}~

The above theorem says that for any $q>\frac{n^{5}}{\log (n)}\cdot C_0$, there exists an MR code $\mathcal{C}$ for topology $T_{4\times n}(1,2,0)$ over $\mathbb{F}_q$. This actually gives an upper bound $\frac{n^{5}}{\log (n)}\cdot C_0$ on the minimal field size required for the existence of an MR code. But is this polynomial trend really necessary? Recall the \emph{MDS~Conjecture}:
\begin{MDS}
If there is a nontrivial $[n,k]$ MDS code over $\mathbb{F}_q$, then $n\leq q+1$, except when $q$ is even and $k=3$ or $k=q-1$ in which case $n\leq q+2$.
\end{MDS}
Since the code $\mathcal{C}_{row}$ is always an MDS code, thus from the \emph{MDS~Conjecture} we know that a linear lower bound is necessary, but will it be sufficient? Sadly not. The next theorem gives a polynomial lower bound on the smallest field size required for the existence of an MR code for the topology $T_{4\times n}(1,2,0)$.

\begin{theorem}\label{lowbT(1,4,0)}
If $q<\frac{(n-3)^{2}}{4}+2$, then for any tensor product code $\mathcal{C}=\mathcal{C}_{col}\otimes\mathcal{C}_{row}$ over $\mathbb{F}_q$ with $\mathcal{C}_{col}$ as a $[4,3,2]$ MDS code and $\mathcal{C}_{row}$ as an $[n,n-2,3]$ MDS code, $\mathcal{C}$ can not be an MR code that instantiates the topology $T_{4\times n}(1,2,0)$.
\end{theorem}

To give the proof, we need the following two propositions:

\begin{proposition}\label{prop2}
Take $\omega\in\mathbb{F}_q^*$ as the primitive element. If there exist six distinct $t_i\in \mathbb{Z}_{q-1}$ such that $t_1+t_6=t_2+t_5=t_3+t_4$, then the polynomial $f(x_1,x_2,\ldots,x_6)=(x_1-x_4)(x_2-x_6)(x_3-x_5)-(x_2-x_4)(x_1-x_5)(x_3-x_6)$ has a zero of the form $(\omega^{t_1},\ldots,\omega^{t_6})$.
\end{proposition}

\begin{proof}
By substituting $(\omega^{t_1},\ldots,\omega^{t_6})$ to $f(x_1,x_2,\ldots,x_6)$ directly, we have
\begin{align*}
&(\omega^{t_1}-\omega^{t_5})(\omega^{t_2}-\omega^{t_4})(\omega^{t_3}-\omega^{t_6})-(\omega^{t_1}-\omega^{t_4})(\omega^{t_2}-\omega^{t_6})(\omega^{t_3}-\omega^{t_5})\\
=& \omega^{t_1+t_2+t_3}[(1-\omega^{t_5-t_1})(1-\omega^{t_4-t_2})(1-\omega^{t_6-t_3})-(1-\omega^{t_6-t_2})(1-\omega^{t_5-t_3})(1-\omega^{t_4-t_1})].
\end{align*}
Since $t_1+t_6=t_2+t_5=t_3+t_4$, then we have
\begin{equation*}
\begin{cases}
t_5-t_1=t_6-t_2\\
t_4-t_2=t_5-t_3\\
t_6-t_3=t_4-t_1
\end{cases}.
\end{equation*}
Using these three identities, we have $f(\omega^{t_1},\ldots,\omega^{t_6})=0$.
\end{proof}

Let $N\geq 2$ be a positive integer, for any subset $A\subseteq \mathbb{Z}_N$, we say $A$ is a \emph{2-Sidon set} if for any $2$-subset $\{a_1,b_1\}\subseteq A$ there exists at most one other $\{a_2,b_2\}\subseteq A$ different from $\{a_1,b_1\}$ such that $a_1+b_1=a_2+b_2$.
\begin{proposition}\label{prop3}
For any $A\subseteq \mathbb{Z}_N$, if $A$ is a \emph{2-Sidon set}, then we have $|A|\leq 2\sqrt{N}+1$.
\end{proposition}

\begin{proof}
Since $A+A\subseteq \mathbb{Z}_N$, by a simple double counting, we have
\begin{equation*}
{|A|\choose 2}\leq 2N.
\end{equation*}
Thus $|A|\leq 2\sqrt{N}+1$.
\end{proof}

\begin{proof}[Proof of Theorem~\ref{lowbT(1,4,0)}]
Different from the proof of the upper bound, since we want to obtain a necessary condition for the existence of an MR code, we have to deal with the general case.

For any $[4,3,2]$ MDS code $\mathcal{C}_1$ and $[n,n-2,3]$ MDS code $\mathcal{C}_2$, take
\begin{equation*}
\mathbf{H}_1=(a_1,a_2,a_3,a_4)~\text{and}~
\mathbf{H}_2=\left(\begin{array}{cccc}
b_{11} & b_{12} & \cdots & b_{1n} \\
b_{21} & b_{22} & \cdots & b_{2n} \\
\end{array}
\right)
\end{equation*}
as their parity check matrices. Then the \emph{pseudo-parity check matrix} of $\mathcal{C}=\mathcal{C}_1\otimes\mathcal{C}_2$ has the following form
\begin{equation*}\label{pmT(1,4,0)}
\mathbf{H}=\left(\begin{array}{cccc}
a_1\cdot\mathbf{I}_n  &  a_2\cdot\mathbf{I}_n  &  a_3\cdot\mathbf{I}_n  &  a_4\cdot\mathbf{I}_n \\ 
\mathbf{H}_{2}  &  \mathbf{0}  &  \mathbf{0}  &  \mathbf{0} \\
\mathbf{0}  &  \mathbf{H}_{2}  &  \mathbf{0}  &  \mathbf{0} \\
\mathbf{0}  &  \mathbf{0}  &  \mathbf{H}_{2}  &  \mathbf{0} \\
\mathbf{0}  &  \mathbf{0}  &  \mathbf{0}  &  \mathbf{H}_{2} \\
\end{array}
\right).
\end{equation*}

Take $\mathbf{H}_1=(a_1,a_2,a_3,a_4)$ as a vector in $\mathbb{F}_q$ and consider its Hamming weight $w(\mathbf{H}_1)$.

%


Since $\mathcal{C}_1$ is a $[4,3,2]$ MDS code, we have $w(\mathbf{H}_1)=4$. Therefore, $a_i\neq 0$ for each $i\in[4]$ and
\begin{equation*}
\mathbf{H}|_{E_s}=\left(\begin{array}{cc}
\mathbf{A'}^{(s)}_{6 \times 6}  &  \mathbf{0}_{6 \times 6}\\[1mm]
\mathbf{A}^{(s)}_{8 \times 6}  &  \mathbf{B}^{(s)}_{8\times 6}\\[1mm]
\mathbf{0}_{(n-6)\times 6}  &  \mathbf{0}_{(n-6)\times 6}\\[1mm]
\end{array}
\right),
\end{equation*}
where
\begin{equation*}
\mathbf{A'}^{(1)}=
\left(\begin{array}{cccccc}
a_1 & & & & & \\
& a_1 & & & & \\
& & a_1 & & & \\
& & & a_2 & & \\
& & & & a_3 & \\
& & & & & a_3\\
\end{array}
\right)~\text{and}~
\mathbf{A'}^{(2)}=
\left(\begin{array}{cccccc}
a_1 & & & & & \\
& a_1 & & & & \\
& & a_1 & & & \\
& & & a_2 & & \\
& & & & a_2 & \\
& & & & & a_3\\
\end{array}
\right),
\end{equation*}
\begin{equation*}
\mathbf{A}^{(1)}=
\left(\begin{array}{cccccc}
\beta_1 & \beta_2 & \beta_3 & & & \\
& & & \beta_4& & \\
& & & & \beta_5 & \beta_6 \\
& & & & & \\
& \multicolumn{4}{c}{\raisebox{1ex}[0pt]{$\mathbf{0}_{2\times 6}$}} &\\
\end{array}
\right)~\text{and}~
\mathbf{A}^{(2)}=
\left(\begin{array}{cccccc}
\beta_1 & \beta_2 & \beta_3 & & & \\
& & & \beta_4 & \beta_5 & \\
& & & & & \beta_6 \\
& & & & & \\
& \multicolumn{4}{c}{\raisebox{1ex}[0pt]{$\mathbf{0}_{2\times 6}$}} &\\
\end{array}
\right),
\end{equation*}
\begin{equation*}
\mathbf{B}^{(1)}=
\left(\begin{array}{cccccc}
-\frac{a_2}{a_1}\beta_1 & -\frac{a_2}{a_1}\beta_2 & -\frac{a_3}{a_1}\beta_3 & & & \\
\beta_1 & \beta_2 & & -\frac{a_4}{a_2}\beta_4 & & \\
& & \beta_3 & & -\frac{a_4}{a_3}\beta_5 & -\frac{a_4}{a_3}\beta_6\\
& & & \beta_4 & \beta_5 & \beta_6\\
\end{array}
\right)
\end{equation*}
and
\begin{equation*}
\mathbf{B}^{(2)}=
\left(\begin{array}{cccccc}
-\frac{a_2}{a_1}\beta_1 & -\frac{a_3}{a_1}\beta_2 & -\frac{a_4}{a_1}\beta_3 & & & \\
\beta_1 & & & -\frac{a_3}{a_2}\beta_4 & -\frac{a_4}{a_2}\beta_5 & \\
& \beta_2 & & \beta_4 & & -\frac{a_4}{a_3}\beta_6\\
& & \beta_3 & & \beta_5 & \beta_6\\
\end{array}
\right),
\end{equation*}
for the column vectors $\{\beta_1,\ldots,\beta_6\}\in \mathbf{H}_2$ corresponding to $E_s$.

It can be easily verified that the first row of $\mathbf{B}^{(s)}$ can be linearly expressed by the other three rows. Thus $rank(\mathbf{B}^{(1)})=rank(\mathbf{B}^{(2)})=6$ if and only if the following two systems of linear equations only have zero solutions.
\begin{equation}\label{leq1}
\begin{cases}
x_1\cdot\beta_1+x_2\cdot\beta_2-x_4\cdot\frac{a_4}{a_2}\beta_4=0  \\
x_3\cdot\beta_3-x_5\cdot\frac{a_4}{a_3}\beta_5-x_6\cdot\frac{a_4}{a_3}\beta_6=0 \\
x_4\cdot\beta_4+x_5\cdot\beta_5+x_6\cdot\beta_6=0 \\
\end{cases}
\end{equation}
\begin{equation}\label{leq2}
\begin{cases}
x_1\cdot\beta_1-x_4\cdot\frac{a_3}{a_2}\beta_4-x_5\cdot\frac{a_4}{a_2}\beta_5=0  \\
x_2\cdot\beta_2+x_4\cdot\beta_4-x_6\cdot\frac{a_4}{a_3}\beta_6=0 \\
x_3\cdot\beta_3+x_5\cdot\beta_5+x_6\cdot\beta_6=0 \\
\end{cases}
\end{equation}

For (\ref{leq1}), if there exists a non-zero solution $(c_1,\ldots,c_6)$, then
\begin{equation*}
c_3\cdot\beta_3=\frac{a_4}{a_3}\cdot (c_5\beta_5+c_6\beta_6)~\text{and}~
c_4\cdot\beta_4=-(c_5\beta_5+c_6\beta_6). \\
\end{equation*}
By this, we have that $\beta_3$ and $\beta_4$ are linearly dependent, which contradicts the MDS property of $\mathbf{H}_2$. Therefore, when $w(\mathbf{H}_1)=4$, despite the size of the field, $rank(\mathbf{B}^{(1)})=6$.

For (\ref{leq2}), different from (\ref{leq1}), it can have non-zero solution $(d_1,\ldots,d_6)$ and does not violate the MDS property of $\mathbf{H}_2$. For example,
take $\omega\in\mathbb{F}_q^*$ as the primitive element, if $\beta_i=(1, \omega^{t_i})$ for some distinct $t_i\in[q-1]$ such that $t_1+t_6=t_2+t_5=t_3+t_4$, then the resulting $\mathbf{B}_2$ has $rank(\mathbf{B}_2)\leq 5$ and this guarantees the existence of non-zero solution for (\ref{leq2}).

W.o.l.g., assume $n\geq 8$, then from the MDS property of $\mathbf{H}_2$, we know that $\mathbf{H}_2$ contains at least $n-2$ weight-$2$ columns. Since any six distinct elements of $[n]$ can be chosen to form an erasure pattern $E_2$ of Type II, therefore, the maximal recoverability requires that $rank(\mathbf{B}^{(2)})=6$ for any six distinct columns in $\mathbf{H}_2$. Especially, we can take all these columns with weight equal to $2$. Assume $\beta_i=(b_{i1},b_{i2})^{T}$ with $b_{is}\neq 0$ for each $i\in[6],s\in[2]$. Then $\mathbf{B}^{(2)}$ can be formulated as
\begin{equation*}
\mathbf{B}^{(2)}=
\left(\begin{array}{@{\hspace{-1pt}}cccccc@{\hspace{-1pt}}}
-\frac{a_2 b_{11}}{a_1}\cdot\left(\begin{array}{@{\hspace{-1pt}}c@{\hspace{-1pt}}} 1 \\ r_1\end{array}\right) &
-\frac{a_3 b_{21}}{a_1}\cdot\left(\begin{array}{@{\hspace{-1pt}}c@{\hspace{-1pt}}} 1 \\ r_2\end{array}\right) &
-\frac{a_4 b_{31}}{a_1}\cdot\left(\begin{array}{@{\hspace{-1pt}}c@{\hspace{-1pt}}} 1 \\ r_3\end{array}\right) & & & \\
b_{11}\cdot\left(\begin{array}{@{\hspace{-1pt}}c@{\hspace{-1pt}}} 1 \\ r_1\end{array}\right) & & &
-\frac{a_3 b_{41}}{a_2}\cdot\left(\begin{array}{@{\hspace{-1pt}}c@{\hspace{-1pt}}} 1 \\ r_4\end{array}\right) &
-\frac{a_4 b_{51}}{a_2}\cdot\left(\begin{array}{@{\hspace{-1pt}}c@{\hspace{-1pt}}} 1 \\ r_5\end{array}\right) & \\
& b_{21}\cdot\left(\begin{array}{@{\hspace{-1pt}}c@{\hspace{-1pt}}} 1 \\ r_2\end{array}\right) & &
b_{41}\cdot\left(\begin{array}{@{\hspace{-1pt}}c@{\hspace{-1pt}}} 1 \\ r_4\end{array}\right) & &
-\frac{a_4 b_{61}}{a_3}\cdot\left(\begin{array}{@{\hspace{-1pt}}c@{\hspace{-1pt}}} 1 \\ r_6\end{array}\right)\\
& &
b_{31}\cdot\left(\begin{array}{@{\hspace{-1pt}}c@{\hspace{-1pt}}} 1 \\ r_3\end{array}\right) & &
b_{51}\cdot\left(\begin{array}{@{\hspace{-1pt}}c@{\hspace{-1pt}}} 1 \\ r_5\end{array}\right) &
b_{61}\cdot\left(\begin{array}{@{\hspace{-1pt}}c@{\hspace{-1pt}}} 1 \\ r_6\end{array}\right)\\
\end{array}
\right),
\end{equation*}
where $r_i=\frac{b_{i2}}{b_{i1}}$ for each $i\in[6]$. Since the first row of $\mathbf{B}^{(2)}$ can be linearly expressed by the other three rows and the scaling of each column doesn't affect the linear dependency, we have $rank(\mathbf{B}^{(2)})=rank(\tilde{\mathbf{B}}^{(2)})$, where
\begin{equation*}
\tilde{\mathbf{B}}^{(2)}=
\left(\begin{array}{@{\hspace{-1pt}}cccccc@{\hspace{-1pt}}}
\left(\begin{array}{@{\hspace{-1pt}}c@{\hspace{-1pt}}} 1 \\ r_1\end{array}\right) & & &
-\frac{a_3}{a_2}\cdot\left(\begin{array}{@{\hspace{-1pt}}c@{\hspace{-1pt}}} 1 \\ r_4\end{array}\right) &
-\frac{a_4}{a_2}\cdot\left(\begin{array}{@{\hspace{-1pt}}c@{\hspace{-1pt}}} 1 \\ r_5\end{array}\right) & \\
&\left(\begin{array}{@{\hspace{-1pt}}c@{\hspace{-1pt}}} 1 \\ r_2\end{array}\right) & &
\left(\begin{array}{@{\hspace{-1pt}}c@{\hspace{-1pt}}} 1 \\ r_4\end{array}\right) & &
-\frac{a_4}{a_3}\cdot\left(\begin{array}{@{\hspace{-1pt}}c@{\hspace{-1pt}}} 1 \\ r_6\end{array}\right)\\
& &
\left(\begin{array}{@{\hspace{-1pt}}c@{\hspace{-1pt}}} 1 \\ r_3\end{array}\right) & &
\left(\begin{array}{@{\hspace{-1pt}}c@{\hspace{-1pt}}} 1 \\ r_5\end{array}\right) &
\left(\begin{array}{@{\hspace{-1pt}}c@{\hspace{-1pt}}} 1 \\ r_6\end{array}\right)\\
\end{array}
\right).
\end{equation*}
And it can be simplified as
\begin{equation*}
\tilde{\mathbf{B}}^{(2)}=\left(\begin{array}{cccccc}
1 & & & & & \\
& 1 & & & & \\
& & 1 & & & \\
& & & \frac{a_3}{a_2}(r_1-r_4) & \frac{a_4}{a_2}(r_1-r_5) & \\
& & & r_4-r_2 & & \frac{a_4}{a_3}(r_2-r_6)\\
& & & & r_5-r_3 & r_6-r_3\\
\end{array}
\right).
\end{equation*}
Therefore, $rank(\mathbf{B}^{(2)})=rank(\tilde{\mathbf{B}}^{(2)})=6$ if and only if $(r_1-r_4)(r_2-r_6)(r_3-r_5)-(r_2-r_4)(r_1-r_5)(r_3-r_6)\neq 0$. \footnote{Recall the condition of $\mathbf{B}^{(2)}$ having full rank and the polynomial $f(x_1,x_2,\ldots,x_6)$ we defined in the proof of Theorem \ref{upbT(1,4,0)}, the condition we obtain here for the general case is actually the same.}

In order to show that the tensor product code $\mathcal{C}$ can't correct all erasure patterns of Type II, we need to prove that if $q$ isn't large enough, there will always be six distinct columns $\{b_{i1}\cdot(1,r_i)^{T}\}_{i\in[6]}$ with $(r_1-r_4)(r_2-r_6)(r_3-r_5)-(r_2-r_4)(r_1-r_5)(r_3-r_6)=0$, which is shown as follows.

Consider $n-2$ distinct columns of $\mathbf{H}_2$ with weight 2, $\{b_{i1}\cdot(1,r_i)^{T}\}_{i\in[n-2]}$, according to the MDS property, we know that $r_i\neq r_j$ for all $i\neq j\in [n-2]$. Therefore, if we take $r_i=\omega^{t_i}$ for each $i\in [n-2]$, we know that $t_i\neq t_j$ for all $i\neq j\in [n-2]$. Denote $A=\{t_i\}_{i\in[n-2]}$, then $A$ is an $(n-2)$-subset of $\mathbb{Z}_{q-2}$. Since $q< \frac{(n-3)^2}{4}+2$, by Proposition \ref{prop3}, we know that $A$ can't be a \emph{2-Sidon set}. Thus, there are at least three different $2$-subsets $\{t_1,t_6\},\{t_2,t_5\},\{t_3,t_4\}\in A$, such that $t_1+t_6=t_2+t_5=t_3+t_4$ and $t_i$s are all distinct. By Proposition \ref{prop2}, the corresponding $\{r_j\}_{j\in[6]}$ such that $r_j=\omega^{t_j}$ for each $j\in [6]$, satisfies $(r_1-r_4)(r_2-r_6)(r_3-r_5)-(r_2-r_4)(r_1-r_5)(r_3-r_6)=0$.

Therefore, $\mathcal{C}$ can not correct the erasure patterns of Type II formed by the corresponding six columns $\{b_{j1}\cdot(1,r_j)^{T}\}_{j\in[6]}$, which means $\mathcal{C}$ is not an MR code that instantiates the topology $T_{4\times n}(1,2,0)$.

\end{proof}

\subsection{MR codes for topologies $T_{3\times n}(1,3,0)$}

First, we will give a complete characterization of the regular irreducible erasure patterns for topology $T_{3\times n}(1,3,0)$.

Denote $\mathcal{E}$ as the set of all the types of regular irreducible erasure patterns for topology $T_{3\times n}(1,3,0)$. For each $E\in\mathcal{E}$, by (\ref{ineq0}), we have $|U_E|+3\leq|V_E|\leq 3|U_E|-3$, which leads to $|U_E|\geq 3$. Since $U_E\subseteq [m]=[3]$, we have $|U_E|=3$ and $|V_E|=6$. Therefore, from (\ref{ineq1}), we have $|E|=12$, and from the irreducibility, each erasure pattern has exactly $2$ erasures in each column and $4$ erasures in each row. Finally by checking the regularity case by case, there is only one type of erasure patterns in $\mathcal{E}$:
\begin{equation*}\label{type0}
E_0=\left(\begin{array}{cccccc}
*  &  *  &  *  &  *  &  \circ  &  \circ  \\[1mm]
*  &  *  &  \circ  &  \circ  &  *  &  *  \\[1mm]
\circ  &  \circ  &  *  &  *  &  *  &  *  \\[1mm]
\end{array}
\right).
\end{equation*}

\subsubsection{Upper bound}~

The following upper bound on the field size required for the existence of an MR code that instantiates the topology $T_{3\times n}(1,3,0)$ also improves the general upper bound from Theorem~\ref{upbT(1,b,0)}.

\begin{theorem}\label{upbT(1,3,0)}
For any $q>\frac{n^{5}}{\log (n)}\cdot C_2$, there exists an MR code $\mathcal{C}$ that instantiates the topology $T_{3\times n}(1,3,0)$ over the field $\mathbb{F}_q$, where $C_2\geq (\frac{10}{c_5})^5$ is an absolute constant.
\end{theorem}

\begin{proof}[Sketch of the proof]
Since the idea of the proof is the same as that of Theorem \ref{upbT(1,4,0)}, we only sketch the main steps here.

Let $\mathcal{C}_{col}$ be the simple parity code $\mathcal{P}_3$, we are going to construct a $3\times n$ matrix $\mathbf{H}_{row}$ such that:
\begin{enumerate}
  \item[(i)] Every $3$ distinct columns of $\mathbf{H}_{row}$ are linearly independent.
  \item[(ii)] For each erasure pattern $E$ of type $E_0$, the \emph{pseudo-parity check matrix} $\mathbf{H}\in\mathbb{F}_q^{(n+9)\times 3n}$ of $\mathcal{P}_3\otimes\mathcal{C}_{row}$ satisfies: $rank(\mathbf{H}|_{E})=12$.
\end{enumerate}

Suppose there exists an objective matrix $\mathbf{A}_{0}$ of the form
\begin{equation*}
\mathbf{A}_{0}=\left(\begin{array}{ccccc}
1 & 1 & 1 & \cdots & 1 \\
a_1 & a_2 & a_3 & \cdots & a_n \\
a_1^2 & a_2^2 & a_3^2 & \cdots & a_n^2 \\
\end{array}
\right),
\end{equation*}
where $\{a_i\}_{i\in[n]}$ are pairwise distinct elements in $\mathbb{F}_q$. Then we have the corresponding \emph{pseudo-parity check matrix}
\begin{equation*}
\mathbf{H_{A_0}}|_{E_0}=\left(\begin{array}{cc}
\mathbf{I}_{6}  &  \mathbf{0}_{6 \times 6}\\[1mm]
\mathbf{A}_{9 \times 6}  &  \mathbf{B}_{9\times 6}\\[1mm]
\mathbf{0}_{(n-6)\times 6}  &  \mathbf{0}_{(n-6)\times 6}\\[1mm]
\end{array}
\right),
\end{equation*}
where
\begin{equation*}
\mathbf{A}=
\left(\begin{array}{cccccc}
1 & 1 & 1 & 1 & & \\
a_1 & a_2 & a_3 & a_4 & & \\
a_1^2 & a_2^2 & a_3^2 & a_4^2 & & \\
& & & & 1 & 1 \\
& & & & a_5 & a_6 \\
& & & & a_5^2 & a_6^2 \\
& & & & & \\
& \multicolumn{4}{c}{\raisebox{1ex}[0pt]{$\mathbf{0}_{3\times 6}$}} &\\
\end{array}
\right)~\text{and}~
\mathbf{B}=
\left(\begin{array}{cccccc}
-1 & -1 & -1 & -1 & & \\
-a_1 & -a_2 & -a_3 & -a_4 & & \\
-a_1^2 & -a_2^2 & -a_3^2 & -a_4^2 & & \\
1 & 1 & & & -1 & -1\\
a_1 & a_2 & & & -a_5 & -a_6 \\
a_1^2 & a_2^2 & & & -a_5^2 & -a_6^2 \\
& & 1 & 1 & 1 & 1\\
& & a_3 & a_4 & a_5 & a_6\\
& & a_3^2 & a_4^2 & a_5^2 & a_6^2\\
\end{array}
\right).
\end{equation*}
Since $\mathbf{B}$ can be simplified as
\begin{equation*}
\mathbf{B}=
\left(\begin{array}{cccccc}
1 & & & & & \\
& 1 & & & & \\
& & a_2-a_1 & & a_1-a_5 & a_1-a_6 \\
& & a_2^2-a_1^2 & & a_1^2-a_5^2 & a_1^2-a_6^2 \\
& & & a_4-a_3 & a_5-a_3 & a_6-a_3 \\
& & & a_4^2-a_3^2 & a_5^2-a_3^2 & a_6^2-a_3^2 \\
& & & & & \\
& \multicolumn{4}{c}{\raisebox{1ex}[0pt]{$\mathbf{0}_{3\times 6}$}} &\\
\end{array}
\right),
\end{equation*}
thus we have
\begin{itemize}
  \item $rank(\mathbf{B})=6$ if and only if $f(a_1,\ldots,a_6)\neq 0$, where \begin{equation*}
      f(x_1,\ldots,x_6)=(x_1-x_2)(x_3-x_4)[(x_1-x_6)(x_2-x_6)(x_3-x_5)(x_4-x_6)-(x_1-x_5)(x_2-x_5)(x_3-x_6)(x_4-x_6)].
      \end{equation*}
\end{itemize}
Let $\mathcal{H}$ be a $6$-uniform hypergraph with vertex set $\mathbb{F}_q$, each set of $6$ vertices $\{v_{1},\ldots,v_{6}\}$ forms a $6$-hyperedge if and only if $f(v_{\pi(1)},\ldots,v_{\pi(6)})=0$ for some $\pi \in S_{6}$. If there exists an independent set $I\subseteq\mathbb{F}_q$ such that $|I|\geq n$, then we can construct an objective matrix $\mathbf{A}_{0}$ by arbitrarily choosing $n$ different vertices from $I$ as $\{a_i\}_{i\in[n]}$.

Since $deg_{x_i}(f)\leq 2$ for each $x_i$, and $f(v_1,\ldots,x_i,\ldots,v_6)$ is a non-zero polynomial for any 5-subset $\{v_{j}\}_{j\in[6]\setminus\{i\}}\in \mathbb{F}_q$, we have $\Delta_{5}(\mathcal{H})\leq 2\cdot6!$. Thus, by Theorem \ref{KMV14},
\begin{equation*}
\alpha(\mathcal{H})\geq \frac{c_5}{5}(q\log{q})^{\frac{1}{5}}>n,
\end{equation*}
therefore, there exists a subset $A=\{a_1,\ldots,a_n\}\subseteq \mathbb{F}_q$ such that the corresponding Vandermonde matrix $\mathbf{A}_0$
is the objective parity check matrix of the row code $\mathcal{C}_{row}$.
\end{proof}

\subsubsection{Lower bound}~

The next theorem gives a linear lower bound on the smallest field size required for the existence of an MR code for topology $T_{3\times n}(1,3,0)$.

\begin{theorem}\label{lowbT(1,3,0)}
If $q<\frac{\sqrt{n^2-11n+34}}{2}$, then for any tensor product code $\mathcal{C}=\mathcal{C}_{col}\otimes\mathcal{C}_{row}$ over $\mathbb{F}_q$ with $\mathcal{C}_{col}$ as a $[3,2,2]$ MDS code and $\mathcal{C}_{row}$ as an $[n,n-3,4]$ MDS code, $\mathcal{C}$ can not be an MR code that instantiates the topology $T_{3\times n}(1,3,0)$.
\end{theorem}

\begin{proof}[Sketch of the proof]
For any $[3,2,2]$ MDS code $\mathcal{C}_1$ and $[n,n-3,4]$ MDS code $\mathcal{C}_2$ , take
\begin{equation*}
\mathbf{H}_1=(a_1,a_2,a_3)~\text{and}~
\mathbf{H}_2=\left(\begin{array}{cccc}
b_{11} & b_{12} & \cdots & b_{1n} \\
b_{21} & b_{22} & \cdots & b_{2n} \\
b_{31} & b_{32} & \cdots & b_{3n} \\
\end{array}
\right)
\end{equation*}
as their parity check matrices. Then the \emph{pseudo-parity check matrix} will be
\begin{equation*}\label{pmT(1,3,0)}
\mathbf{H}=\left(\begin{array}{ccc}
a_1\cdot\mathbf{I}_n  &  a_2\cdot\mathbf{I}_n  &  a_3\cdot\mathbf{I}_n \\ 
\mathbf{H}_{2}  &  \mathbf{0}  &  \mathbf{0}  \\
\mathbf{0}  &  \mathbf{H}_{2}  &  \mathbf{0}  \\
\mathbf{0}  &  \mathbf{0}  &  \mathbf{H}_{2}  \\
\end{array}
\right).
\end{equation*}

Take $\mathbf{H}_1=(a_1,a_2,a_3)$ as a vector in $\mathbb{F}_q$ and consider its Hamming weight $w(\mathbf{H}_1)$. Similarly, from the MDS property of $\mathcal{C}_1$, we have $w(\mathbf{H}_1)=3$. Thus we have
\begin{equation*}
\mathbf{H}|_{E_0}=\left(\begin{array}{cc}
\mathbf{A'}_{6 \times 6}  &  \mathbf{0}_{6 \times 6}\\[1mm]
\mathbf{A}_{9 \times 6}  &  \mathbf{B}_{9\times 6}\\[1mm]
\mathbf{0}_{(n-6)\times 6}  &  \mathbf{0}_{(n-6)\times 6}\\[1mm]
\end{array}
\right),
\end{equation*}
where
\begin{equation*}
\mathbf{A'}=
\left(\begin{array}{@{\hspace{-1pt}}cccccc@{\hspace{-1pt}}}
a_1 & & & & & \\
& a_1 & & & & \\
& & a_1 & & & \\
& & & a_1 & & \\
& & & & a_2 & \\
& & & & & a_2\\
\end{array}
\right)~\text{and}~
\mathbf{B}=
\left(\begin{array}{@{\hspace{-1pt}}cccccc@{\hspace{-1pt}}}
-\frac{a_2}{a_1}\beta_1 & -\frac{a_2}{a_1}\beta_2 & -\frac{a_3}{a_1}\beta_3 & -\frac{a_3}{a_1}\beta_4 & & \\
\beta_1 & \beta_2 & & & -\frac{a_3}{a_2}\beta_5 & -\frac{a_3}{a_2}\beta_6\\
& & \beta_3 & \beta_4 & \beta_5 & \beta_6\\
\end{array}
\right),
\end{equation*}
for the column vectors $\{\beta_1,\ldots,\beta_6\}\in \mathbf{H}_2$ corresponding to $E_0$.

Different from the case for topology $T_{4\times n}(1,2,0)$, first, we shall divide the columns of $\mathbf{H}_2$ into the following two parts.
\begin{itemize}
  \item Part one: $(b_{1i},b_{2i},b_{3i})^{T}\in \mathbf{H}_2$ with $b_{1i}\neq0$;
  \item Part two: $(b_{1i},b_{2i},b_{3i})^{T}\in \mathbf{H}_2$ with $b_{1i}=0$.
\end{itemize}

Denote the set of columns in Part two as $C_2$, if $|C_2|\geq 6$, we can choose six distinct column vectors from $C_2$ as the corresponding $\{\beta_1,\ldots,\beta_6\}$ in $\mathbf{B}$. Since $\beta_j(1)=0$ for each $j\in[6]$, we have $rank(\mathbf{B})\leq 4$. Thus $\mathcal{C}$ can not correct all erasure patterns of type $E_0$ if $|C_2|\geq 6$.

Assume $|C_2|\leq 5$. Denote the set of columns in Part one as $C_1$. Assume $|C_1|\geq 6$, take any six distinct column vectors from $C_1$ as the corresponding $\{\beta_1,\ldots,\beta_6\}$ in $\mathbf{B}$. Similarly, after some scaling process for columns in $\mathbf{B}$, we can get a matrix $\tilde{\mathbf{B}}$ of the form
\begin{equation*}
\tilde{\mathbf{B}}=
\left(\begin{array}{@{\hspace{-1pt}}cccccc@{\hspace{-1pt}}}
\left(\begin{array}{@{\hspace{-1pt}}c@{\hspace{-1pt}}} 1 \\ r_{11} \\ r_{12}\end{array}\right) &
\left(\begin{array}{@{\hspace{-1pt}}c@{\hspace{-1pt}}} 1 \\ r_{21} \\ r_{22}\end{array}\right) & & &
-\frac{a_3}{a_2}\cdot\left(\begin{array}{@{\hspace{-1pt}}c@{\hspace{-1pt}}} 1 \\ r_{51} \\ r_{52}\end{array}\right) &
-\frac{a_3}{a_2}\cdot\left(\begin{array}{@{\hspace{-1pt}}c@{\hspace{-1pt}}} 1 \\ r_{61} \\ r_{62}\end{array}\right) \\
& &
\left(\begin{array}{@{\hspace{-1pt}}c@{\hspace{-1pt}}} 1 \\ r_{31} \\ r_{32}\end{array}\right) &
\left(\begin{array}{@{\hspace{-1pt}}c@{\hspace{-1pt}}} 1 \\ r_{41} \\ r_{42}\end{array}\right) &
\left(\begin{array}{@{\hspace{-1pt}}c@{\hspace{-1pt}}} 1 \\ r_{51} \\ r_{52}\end{array}\right) &
\left(\begin{array}{@{\hspace{-1pt}}c@{\hspace{-1pt}}} 1 \\ r_{61} \\ r_{62}\end{array}\right)\\
\end{array}
\right),
\end{equation*}
such that $rank(\mathbf{B})=rank(\tilde{\mathbf{B}})$, where $\beta_j=\beta_j(1)\cdot(1,r_{j1},r_{j2})^T$ for each $j\in[6]$. And $\tilde{\mathbf{B}}$ can be simplified as
\begin{equation*}
\tilde{\mathbf{B}}=\left(\begin{array}{cccccc}
1 & & & & & \\
& 1 & & & & \\
& & \gamma_2-\gamma_1 & & \frac{a_3}{a_2}(\gamma_1-\gamma_5) & \frac{a_3}{a_2}(\gamma_1-\gamma_6) \\
& & & \gamma_4-\gamma_3 & \gamma_5-\gamma_3 & \gamma_6-\gamma_3 \\
\end{array}
\right),
\end{equation*}
where $\gamma_j=(r_{j1},r_{j2})^T$ for each $j\in[6]$.

In order to show that the tensor product code $\mathcal{C}$ can't correct all erasure patterns of type $E_0$, we need to prove that if $q$ isn't large enough, there will always be six distinct columns $\{b_{j1}\cdot(1,\gamma_j)^{T}\}_{j\in[6]}$ in $C_1$ with the resulting $\tilde{\mathbf{B}}$ having rank less than 6.

Let $A\subseteq \mathbb{F}_q^2$ such that for any $\{a_1,a_2\}\subseteq A$, there exists at most one other 2-subset $\{a_3,a_4\}\subseteq A$ different from $\{a_1,a_2\}$ satisfying $a_2-a_1=a_4-a_3$.\footnote{This subset $A$ here can be viewed as a generalized \emph{2-Sidon} set over the vector space.}
Thus we have
\begin{equation*}
{|A|\choose 2}\leq 2\cdot(q^2-1),
\end{equation*}
and $|A|\leq 2\sqrt{q^2-\frac{15}{16}}+\frac{1}{2}$.

Since $q<\frac{\sqrt{n^2-11n+34}}{2}$, we have $|C_1|=n-|C_2|\geq n-5>2\sqrt{q^2-\frac{15}{16}}+\frac{1}{2}$. Therefore, there exist six distinct columns $\{b_{j1}\cdot(1,\gamma_j)^{T}\}_{j\in[6]}$ in $C_1$ such that $\gamma_2-\gamma_1=\gamma_4-\gamma_3=\gamma_6-\gamma_5$ and the matrix $\tilde{\mathbf{B}}$ corresponding to $\{b_{j1}\cdot(1,\gamma_j)^{T}\}_{j\in[6]}$ has $rank(\tilde{\mathbf{B}})=5<6$. Thus $\mathcal{C}$ can not correct all erasure patterns of type $E_0$.

\end{proof}

\section{Concluding Remarks and Further Research}

In this paper, we obtain a polynomial upper bound on the minimal size of fields required for the existence of MR codes that instantiate the topology $T_{m\times n}(1,b,0)$, which improves the general upper bound given by Gopalan et al. \cite{Gopalan2014}. We also consider some special cases with fixed $m$ and $b$, for each of which, we obtain an improved upper bound and a non-trivial lower bound. Though many works have been done, there is still a wide range of questions that remain open. Here we highlight some of the questions related to our work.

\begin{itemize}
  \item Due to the rough estimation on the number of regular erasure patterns, the upper bound given by Theorem \ref{upbT(1,b,0)} still grows exponentially with $m$. If one can give a better characterization of the regular erasure patterns (probably using tools from extremal graph theory), we believe the general upper bound in Theorem \ref{upbT(1,b,0)} can also be improved.
  \item As for the lower bound, we only considered the two simplest cases. For general case, due to the complexity of the erasure patterns, our method might not work. Therefore, a general non-trivial lower bound on the field size of codes achieving the MR property for topologies $T_{m\times n}(1,b,0)$ remains widely open.
  \item Under the limitations of the methods themselves, the Combinatorial Nullstellensatz and the hypergraph independent set approach can only give existence results. Therefore, explicit constructions of MR codes for topologies $T_{m\times n}(1,b,0)$ over small fields are still interesting open problems. In particular, is it possible to give an explicit construction of MR codes for the topology $T_{4\times n}(1,2,0)$ over a field of size between $\Omega(n^2)$ and $\mathcal{O}(n^5/\log (n))$?
  \item Unfortunately, the lower bound given by Theorem \ref{lowbT(1,3,0)} doesn't beat the lower bound $q\geq n-1$ or $q\geq n-2$ given by the \emph{MDS conjecture}. But considering the matrix $\tilde{\mathbf{B}}$ in the proof, the MDS property already ensures $\tilde{\mathbf{B}}$ has full rank when restricting to most of the columns. Thus, probably the upper bound given by Theorem \ref{upbT(1,3,0)} can be improved to be a linear function of $n$.
\end{itemize}


%

\end{document}